%% file: paper.tex
\documentclass[journal,comsoc]{IEEEtran}
\usepackage[T1]{fontenc}

\usepackage{booktabs}
\usepackage{amsmath,amssymb,amsthm}
\usepackage[cmintegrals]{newtxmath}
\usepackage{todonotes}
\usepackage{subfigure}
\usepackage{graphicx}
\usepackage[noend]{algpseudocode}
\usepackage{algorithm}
\usepackage{multirow}
\usepackage{lipsum}
\usepackage[bookmarks=false]{hyperref}

\newtheorem{theorem}{Theorem}
\graphicspath{{./}}

\algrenewcommand{\algorithmiccomment}[1]{/* \textit{#1} */}

\begin{document}

\title{JASPER: Joint Optimization of Scaling, Placement, and Routing of Virtual Network Services}

%

\author{Sevil~Dr\"axler,
        Holger~Karl,
        and~Zolt\'an~\'Ad\'am~Mann
\thanks{S. Dr\"axler and H. Karl are with Paderborn University, Paderborn, Germany.}
\thanks{Z. A. Mann is with University of Duisburg-Essen, Essen, Germany.}
\thanks{This work has been submitted to the IEEE for possible publication.  Copyright may be transferred without notice, after which this version may no longer be accessible.}
}

\markboth{Manuscript Under Review}%
{Dr\"axler \MakeLowercase{\textit{et al.}}: JASPER: Joint Optimization of Scaling, Placement, and Routing of Virtual Network Services}
%


\maketitle

\input{abstract}

\input{intro}
\input{previous}
\input{approach}
\input{problem}

\input{complexity}
\input{integer}
\input{heuristic}

\input{evaluation}
\input{conclusion}

\section*{Acknowledgment}

This work has been performed in the context of the SONATA project, 
funded by the European Commission under Grant number 671517 through 
the Horizon 2020 and 5G-PPP programs. This work is partially supported 
by the German Research Foundation (DFG) within the Collaborative Research 
Center ``On-The-Fly Computing'' (SFB 901).

The work of Z.\ \'A.\ Mann was partially supported by the Hungarian Scientific Research Fund (Grant Nr. OTKA 108947) and the European Union's
Horizon 2020 research and innovation programme under grant 731678 (RestAssured).


\bibliographystyle{IEEEtran}
\bibliography{ref}

\vfill


\end{document}

%% file: abstract.tex
\begin{abstract}



To adapt to continuously changing workloads in networks, components of the running network services
may need to be replicated (\emph{scaling} the network service) and allocated
to physical resources (\emph{placement}) dynamically, also
necessitating dynamic re-routing of flows between service
components. In this paper, we propose JASPER, a fully automated
approach to jointly optimizing scaling, placement, and routing for
complex network services, consisting of multiple (virtualized) components. JASPER handles multiple network services
that share the same substrate network; services can be dynamically
added or removed and dynamic workload changes are
handled.
Our approach lets service designers specify their services on a high
level of abstraction using \emph{service templates}. From the
service templates and a description
of the substrate network, JASPER automatically
makes scaling, placement and routing decisions, enabling quick
reaction to changes. We formalize the problem, analyze its
complexity, and develop two algorithms to solve it. Extensive
empirical results show the applicability and effectiveness of the
proposed approach.


\end{abstract}


%% file: intro.tex

\section{Introduction}

Network services, like video streaming and online gaming, consist of
different service components, including (virtual) network functions,
application servers, data bases, etc. Typically, several of these
network services are hosted on top of wide-area
networks, 
serving the continuously
changing demands of their users. The need for efficient and automatic
deployment, scaling, and path selection methods for the network
services has led to paradigms like network softwarization, including software-defined networking (SDN) and network
function virtualization (NFV).


SDN and NFV provide the required control and
orchestration
mechanisms to drive the
network services through their life-cycle. Today, network services are placed
and deployed in the network based on fixed, pre-defined
descriptors~\cite{etsi-mano} that contain the number of required
instances for each service component and the exact resource
demands. More flexibility can be achieved by specifying auto-scaling
thresholds for metrics of interest. Once such a threshold is reached,
the affected network services should be  modified, e.g., scaled.

To react to addition and removal of network services, fluctuations in the request load of a network service, or to serve new user groups in a new location, (i) the network services can be scaled out/in by adding/removing instances of service components, (ii) the placement of service components and the amount of resources allocated to them can be modified, and (iii) the network flows between the service components can be re-routed through different, more suitable paths.

Given this large number of degrees of freedom for finding the best adaptation, deciding scaling, placement, and routing independently can result in sub-optimal decisions for the network and the running services. 
Consider a service platform provider hosting a dynamically changing
set of network services, where each network service serves
dynamically changing user groups that produce dynamically changing
data rates.
Trade-offs among the conflicting goals of network services and
platform operators
can be highly non-trivial, for example:
\begin{itemize}
\item Placing a compute-intensive service component on a node with limited resources near the
  \emph{source} of requests (e.g., the location of users, content
  servers, etc.) minimizes latency but placing it on a more powerful
  node further away in the network minimizes processing time.
\item Letting a single instance of a data-processing component serve
  multiple sources minimizes compute resource consumption but  using
  dedicated instances near the sources  minimizes network load.
\item Changing the current configuration to a better one will
  hopefully pay off in the long run but keeping the current
  configuration avoids reconfiguration costs.
\item Fulfilling the resource requirements of one service versus the requirements of another service.
\end{itemize}

To deal with these challenges, we propose JASPER, a comprehensive approach for the \textbf{J}oint optimiz\textbf{A}tion of \textbf{S}caling, \textbf{P}lac\textbf{E}ment, and \textbf{R}outing of virtual network services.
In JASPER, each network service is described by a \emph{service template}, containing information about the components of the network service, the interconnections between the components, and the resource requirements of the components. Both the resource requirements and the outgoing data rates of a component are specified as \emph{functions of the incoming data rates}. 

The input to the problem we are tackling comprises service templates, location and data rate of the \emph{sources} of each network service, and the topology and available resources of the underlying network. 
Our optimization approach takes care of the rest: based on the
location and current data rate of the sources,
in a single step, the templates are scaled by replicating service components as necessary, the placement of components on network nodes is determined, and data flows are routed along network paths. Node and link capacity constraints of the network are automatically taken into account. We optimize the solution along multiple objectives, including minimizing resource usage, minimizing latency, and minimizing deployment adaptation costs.

Our main contributions are  as follows:
\begin{itemize}
    \item For the case where resource demands of service components are
    determined as a function of the incoming data rate to each
    instance, we formalize \emph{template embedding} as a joint optimization problem for scaling, placing, and routing service templates in the network.
    \item We prove the NP-hardness of the problem.
    \item We present two algorithms for solving the problem, one based on mixed integer programming, the other a custom heuristic.
	\item We evaluate both algorithms in detail to determine their strengths and weaknesses.
\end{itemize}

With the proposed approach, service providers obtain a flexible way to define network services on a high level of abstraction while service platform providers obtain powerful methods to optimize the scaling and placement of multiple services in a single step, fully automatically.

The rest of the paper is organized as follows. In Section~\ref{sec:previous}, we
give an overview of related work. Section~\ref{sec:approach} presents a high-level overview of our approach and Section~\ref{sec:problem} describes the details of our model and assumptions. We discuss the complexity of template embedding
in Section~\ref{sec:compl} and formulate the problem as a mixed integer
programming model in Section~\ref{sec:integer}. We present a heuristic solution 
in Section~\ref{sec:heur} and the evaluation results of our solutions in 
Section~\ref{sec:eval}, before concluding the paper in Section~\ref{sec:concl}.


%% file: previous.tex

\section{Related work}
\label{sec:previous}

The template embedding problem is a joint, single-step optimization of scaling, placement, and routing for network services. In general, our solution can be applied in different contexts, e.g., (distributed) cloud computing and Network Function Virtualization (NFV). In this section, after an analysis of related approaches from a theoretical point of view, we give an overview of related work in the cloud computing and NFV contexts. The major difference among the existing work in these two fields is usually the abstraction level considered for the substrate network and the resulting assumptions for the model. In particular, in the cloud computing context, embedding is typically done on top of physical machines in data centers, while in the NFV context,  embedding is done on top of geographically distributed points of presence.  

\subsection{Virtual network embedding problem}

The combination of the placement and path selection sub-problems of
template embedding is similar to the Virtual Network Embedding (VNE)
problem. Both deal with mapping virtual nodes and virtual links of a graph into another graph and do not include the scaling step. Fischer et al.~\cite{Fischer2013} have published a survey of different approaches to VNE, including static and dynamic VNE algorithms. 

In contrast to static VNE solutions that consider the initial mapping process only, in this paper we also deal with optimizing and modifying already embedded templates. Some VNE solutions, for example, Houidi et al.~\cite{houidi2015exact}, can modify the mapping in reaction to node or link failures. The modifications in their work, however, are limited to recalculating the location for the embedded virtual network, i.e., migrating some of the nodes and changing the corresponding paths among them. In addition to these modifications, our approach can also modify the \emph{structure} of the graph to be embedded by adding or removing nodes and links if necessary.

\subsection{Cloud computing context}

The related problem in cloud environments is typically formulated as resource allocation for individual components. Scaling and placing instances of virtual machines on top of physical machines while adhering to capacity constraints are the usual problems tackled in this context~\cite{lorido2014review,mann2016interplay}. The communication among different virtual machines, however, is usually left out or considered only in a limited sense~\cite{mann2015allocation}. Even the approaches that do consider the communication among virtual machines~\cite{divakaran2015towards,ahvar2015nacer,alicherry2013optimizing,ahvar2016cacev} do not include routing decisions whereas JASPER also includes routing.

Relevant to the placement sub-problem of template embedding, Bellavista et al.~\cite{Bellavista2015} focus on the technical issues of deploying flexible cloud infrastructure, including network-aware placement of multiple virtual machines in virtual data centers. Wang et al.~\cite{Wang2017} study the dynamic scaling and placement problem for network services in cloud data centers, aiming at reducing costs. These papers also do not address routing. Moreover, our approach of specifying resource consumption as a function of input data rates allows a much more realistic modeling of the resource needs of service components than the constant resource needs assumed by the existing approaches in this context.

Keller et al.~\cite{Keller2014b} consider an approach similar to our template embedding problem in the context of distributed cloud computing. Our terminology is partly based on their work but there are important differences in the assumptions and the models that make our approach stronger and more flexible than their solutions. In contrast to their model, where the number of users determines the number of required instances, the deciding factor in our work is the \emph{data rate} originating from different source components. Data rate can be represented, for example, as requests or bits per second and is a more perceptible metric in practical applications and gives a more fine-grained control over the embedding process. Moreover, we do not enforce strict scaling restrictions for components as  done in their work. (For example, their method needs as input the exact number of instances of a back-end server that is required behind a front-end server.) Finally, the optimization objective in their model is limited to minimizing the total number of instances for embedded templates. We use a more sophisticated multi-objective optimization approach where different metrics like CPU and memory load of network nodes, data rate on network links, and latency of embedded templates are considered.

\subsection{Network function virtualization context}

The placement and routing problems are also relevant in the field of Network Function Virtualization (NFV). In the NFV context, the \emph{forwarding graphs} of network services composed of multiple virtual network functions (VNFs) are mapped into the network. Herrera et al.~\cite{herrera2016resource} have published an analysis of existing solutions for placing network services as part of a survey on resource allocation in NFV. 

Kuo et al.~\cite{Kuo2016} consider the joint placement and routing problem, focusing on maximizing the number of admitted network service embedding requests. Ahvar et al.~\cite{Ahvar2017} propose a solution to this problem, with the assumption that the VNFs can be re-used among different flows. Their objective is to find the optimal number of VNFs for all requests and to minimize the costs for the provider. Another similar approach that considers re-using components is proposed by Bari et al.~\cite{bari2016orchestrating}. 
Kebbache et al.~\cite{Khebbache2017} aim at solving this problem in an
efficient way that can scale with the size of the underlying
infrastructure and the embedded network services. They measure the
efficiency of their algorithms with respect to run time, acceptance
rate, and costs. Another attempt to solve this problem in an efficient
and scalable way has been made by Luizelli et al.~\cite{Luizelli2017},
focusing on minimizing resource allocation. In comparison to all these approaches, we consider a more comprehensive optimization objective, trying to minimize the delay for  network services, the number of added or removed instances, resource consumption, as well as overload of resources.

In our work, the exact structure of the network service does not
have to be fixed in the deployment request.
In a previous work~\cite{draexler2017ijnm}, we have studied another type
of flexibility in the network service structure, namely, the case
where the network service components are specified with a partial
order and can be re-ordered if desirable for the optimization
objectives. Beck et al.~\cite{Beck2015} also consider placement of network services with flexibly ordered components. JASPER is based on the assumption that the \emph{order} of traversing the service components is fixed and given, however, the number of instances for each component and the amount of resources allocated to each component can be adapted dynamically, resulting in network services with malleable structures.

Several other optimization approaches~\cite{mehraghdam-netsoft16,sahhaf2015,moens2014vnf,savi2015impact} and heuristic 
algorithms~\cite{mijumbidesign,beck2015coordinated} have been proposed for placement, scaling, and path selection problems for network services. Our template embedding approach has two important differences compared to these
solutions. First, our approach can be used for initial placement of a newly requested service as well as 
scaling and adapting existing embeddings. Second, in our approach, the structure of the service and mapping of the service components to network nodes and the optimal routing are determined in one 
single step, based on the requirements of the service and current state
of network resources, searching for a global optimum.

A preliminary version of this work was presented at the CCGrid 2017
conference \cite{draexler2017joint}. Compared to the conference
version, this paper contains the proof of NP-hardness, more detailed
explanation of the problem model and the devised algorithms, and a
more detailed evaluation and discussion of the practical applicability
of the proposed approach.


%% file: approach.tex

\section{Approach Overview}
\label{sec:approach}

In typical management and orchestration frameworks~\cite{etsi-mano},
service providers need to submit exact descriptors of their network
service structure, resource demands, and expected traffic from sources to a service management and orchestration system
(Fig.~\ref{fig:orch-normal}).
Based on the descriptors, placement,
scaling, and routing decisions are made for each network service, independently from one another.


\begin{figure}[t]
\centering
\includegraphics[width=0.9\columnwidth]{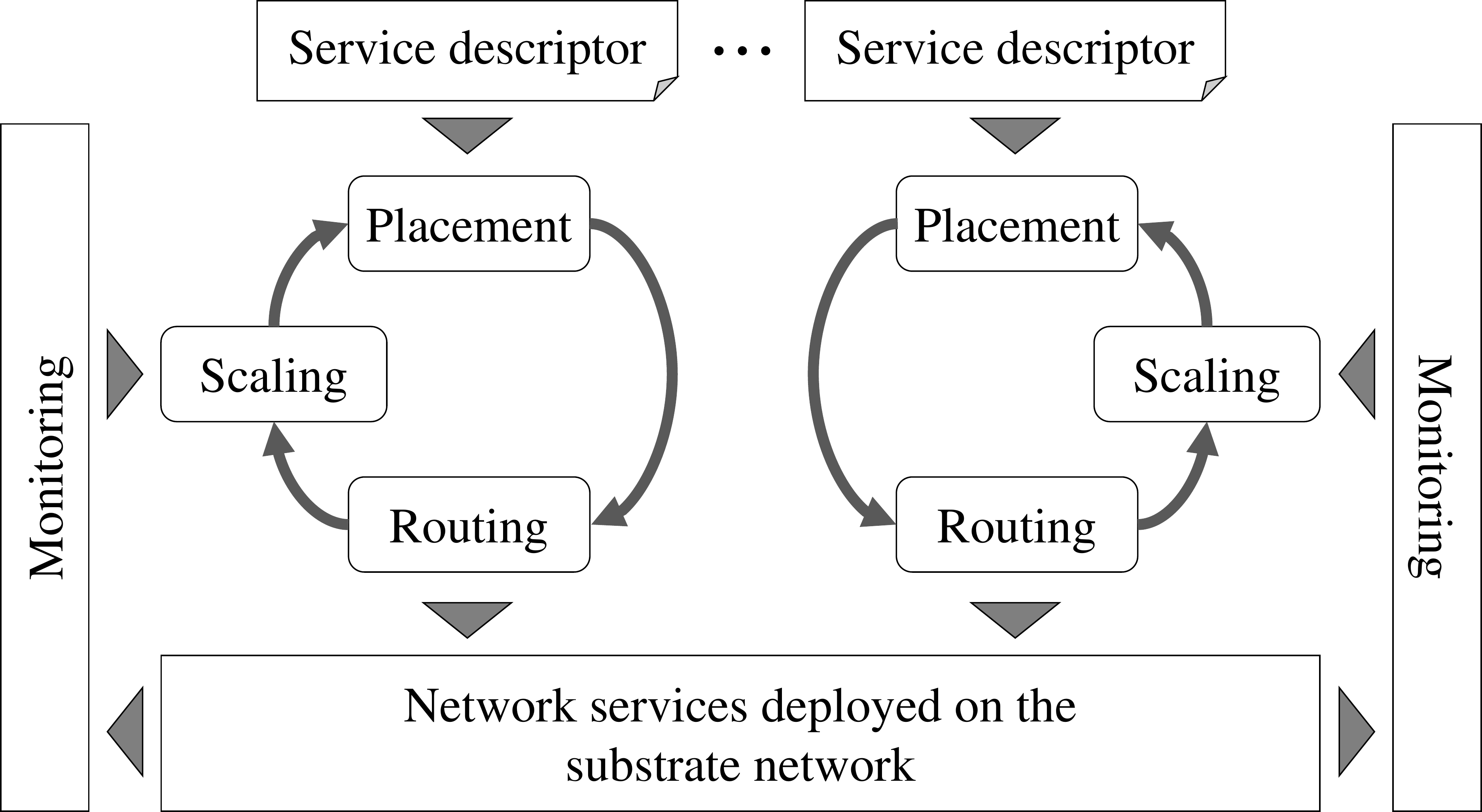}
\caption{Conventional network service life-cycle, from descriptors to
  running services}
\label{fig:orch-normal}
\end{figure}


\begin{figure}[t]
\centering
\includegraphics[width=0.8\columnwidth]{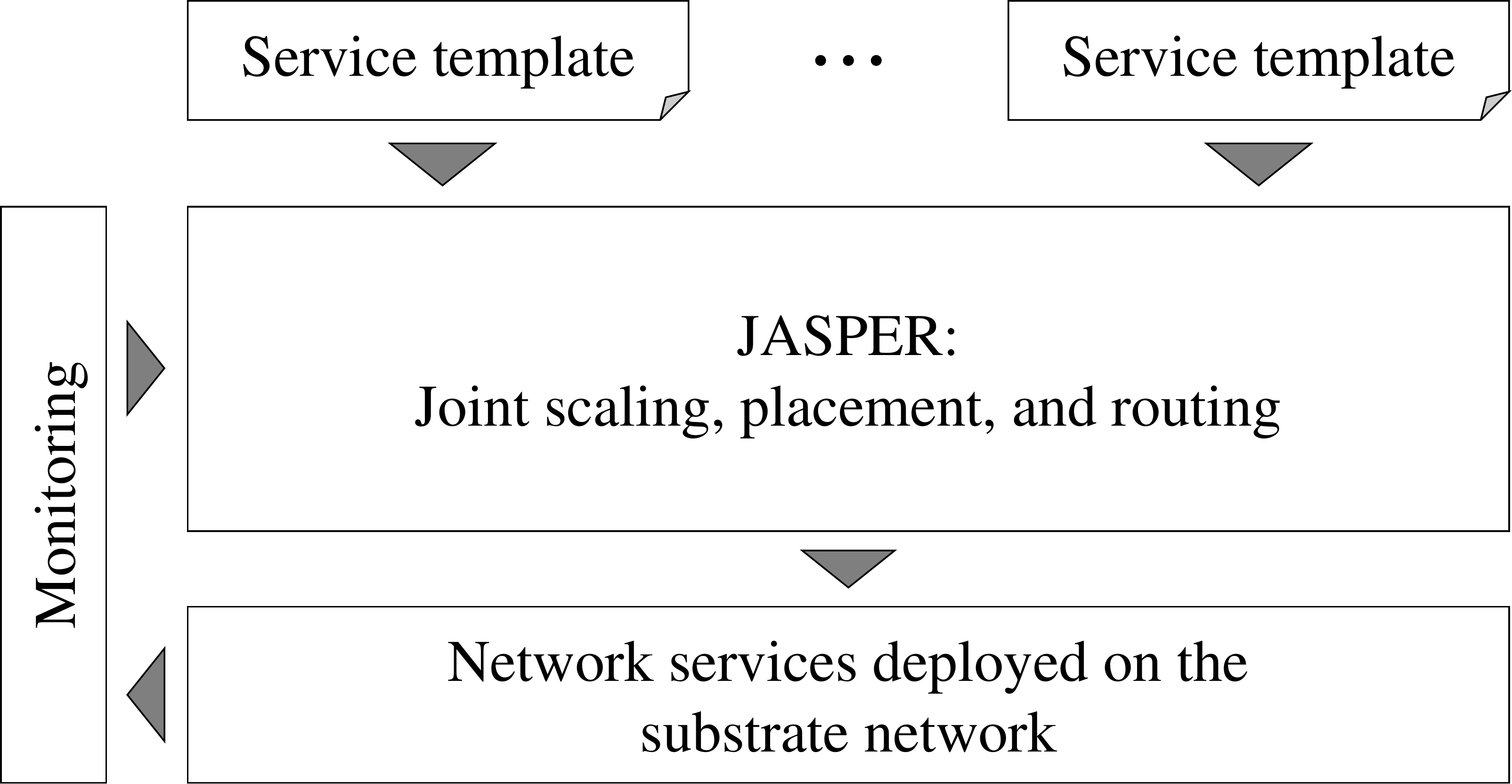}
\caption{Network service management and orchestration using JASPER}
\label{fig:orch-jasper}
\end{figure}

JASPER makes two major changes to this approach, one with respect to
the description of the network services
(Section~\ref{sec:templ-inst-over}) and another with respect to
handling the scaling, placement, and routing decision processes (Section~\ref{sec:joint-single-step}).

\subsection{Templates instead of over-specified descriptors}
\label{sec:templ-inst-over}

Because of the limited precision and flexibility of typical
descriptors, we base our approach on so-called \emph{service
  templates}. Using service templates, service providers are required
to specify neither the \emph{exact} resource demands (e.g., memory or
CPU) of service components nor the required number of instances of each component. 

The service template describes the components of the network service and their required interconnections on an abstract level, without deployment details. Moreover, it gives the resource demands of the network service as a function of the load:
\begin{itemize}
    \item The required computational capacity (e.g., CPU and memory) is described for each service component as a function of the input data rate. This can be used to calculate the network node capacity required to host the service component.
    \item The amount of traffic leaving each service component towards other components is specified as a function of the data rate that enters the component. This can be used to calculate the link capacity required to host the traffic flowing between any two interconnected instance.
\end{itemize}

In addition to the service templates, service providers may include the expected traffic originating from the sources of the network service in the request to embed a service template. As the traffic is constantly changing, the current traffic needs to be monitored and fed back to the template embedding process, to keep the network service in an optimal state. 
In this way, depending on the location and data rate of the sources of the network service, resource requirements are calculated dynamically, based on the given functions, eliminating the risk of over- or under-estimating the resource demands. Based on the functions describing the dependency of resource requirements and outgoing data rates on incoming data rates, it is also possible to reason about possible changes to the deployment and their impact, which is a pre-requisite for effective optimization.
The specific functions highly depend on the type and implementation of the service component and can be derived, for example, based on historical usage data or by automatic service profiling methods~\cite{profiling2016}. 


\subsection{Joint, single-step scaling, placement, and routing}
\label{sec:joint-single-step}

As shown in Fig.~\ref{fig:orch-normal}, in typical management and
orchestration frameworks~\cite{etsi-mano}, based on the description of
the network service and the state of the network's resources, the
requested number of instances for each service component are computed
and then placed and deployed with the requested amount of resources,
in an appropriate location. After path selection and instantiation of
the network service, the running instances are monitored and re-scaled
and re-placed based on pre-defined scaling rules if required. 
Deciding the number of required instances for each service component,
the amount of resources allocated to each component, and the optimal
paths selected for routing the network service flows are, however,
highly interdependent problems, which cannot be solved optimally using
such independent management and orchestration steps. 

Our approach, illustrated in Fig.~\ref{fig:orch-jasper}, changes the way network service life-cycle is handled, by combining scaling, placement, and routing steps into a joint decision process.  
Depending on the location and data rate of the sources,
\begin{itemize}
    \item each service template is scaled out into an overlay with the necessary number of instances required for each service component;
    \item each component instance is mapped to a network node and is allocated the required amount of resources on that node;
    \item the connections among component instances are mapped to flows along network links, carrying the data rate.
\end{itemize}
JASPER is an integrated approach in multiple dimensions: (i) scaling, placement, and routing decisions are made in a single optimization step; (ii) all services that are to be placed in the same substrate network are considered together; (iii) newly requested and already deployed services are optimized jointly. This way, a global optimum can be achieved.

Modern management and orchestration systems~\cite{osmwebsite,sonatawebsite,unifywebsite} have a flexible design to incorporate innovative life-cycle management approaches. For example, SONATA's service platform~\cite{sonata-paper} has a customizable service life-cycle management plugin. The platform operator can easily modify the order of life-cycle management operations and customize different operations. Using service-specific management programs it is also possible to specify when and how scaling, placement, and routing operations are performed for each network service, making the practical implementation of JASPER possible.


%% file: problem.tex

\section{Problem model}
\label{sec:problem}

In this section, we formalize our model and define the problem we are tackling.
Our model uses three different graphs for representing (i) the generic
network service structure, (ii) a concrete and deployable
instantiation of the network service, and (iii) the actual network. We use different names and notations to distinguish among these graphs (Table~\ref{tab:graphs}).

Informally, the problem we  address is as follows: given a substrate network, a set of -- newly requested or already existing -- network services with 
their templates, and the source(s) for the services in the network
along with the traffic originating from them, we want to optimally embed the network services into the network.

\begin{table}[t]
\centering
\caption{Notations Used for Graphs in the Model}
\label{tab:graphs}
\begin{tabular}{llll}
\toprule
Graph                                   & Symbol          & Name    & Annotations \\ 
\midrule
\multirow{2}{*}{Template $G_\text{tmpl}$}                & $j{\in} C_T$         & Component & $\text{In}(j)$, $\text{Out}(j)$, $p_j,m_j,r_j$  \\
                                        & $a{\in} A_T$         & Arc       \\ 
\midrule
\multirow{2}{*}{Overlay $G_\text{OL}$}                   & $i{\in} I_\text{OL}$ & Instance & $c(i)$, $P_T^{(I)}(i)$  \\
                                        & $e{\in} E_\text{OL}$ & Edge & $P_T^{(E)}(e)$     \\
\midrule
\multirow{2}{*}{Network $G_\text{sub}$} & $v{\in} V$           & Node  &  $\text{cap}_{\text{cpu}}(v)$, $\text{cap}_{\text{mem}}(v)$  \\
                                        & $l{\in} L$           & Link  & $b(l)$, $d(l)$    \\ 
\bottomrule
\end{tabular}
\end{table}

\subsection{Substrate network}

We model the \emph{substrate network} as a directed graph
$G_\text{sub}{=}(V,L)$. Each \emph{node} $v\in V$ is associated with a
CPU capacity $\text{cap}_{\text{cpu}}(v)$ and a memory capacity
$\text{cap}_{\text{mem}}(v)$ (this can be easily extended to other
types of resources). Moreover, we assume that every node has routing
capabilities and can forward traffic to its neighboring
nodes.\footnote{Capacities can be 0, e.g., to represent conventional
  switches by 0 CPU capacity or an end device by 0 forwarding capacity.} 
Each \emph{link} $l\in L$ is associated with a maximum data rate
$b(l)$ and a propagation delay $d(l)$. For each node $v$, we assume that the internal communications (e.g., communication inside a data center) can be done with unlimited data rate and negligible delay.

\subsection{Templates}

The substrate network has to host a set $\cal T$ of network services. We define the structure of each network service $T\in {\cal T}$ using a \emph{template}, which is a directed acyclic graph $G_\text{tmpl}(T){=}(C_T,A_T)$. We refer to the nodes and edges of the template graph as \emph{components} and \emph{arcs}, respectively. They define the type of components required in the network service and specify the way they should be connected to each other to deliver the desired functionality.
Fig.~\ref{fig:template} shows an example template. 

A template component $j\in C_T$ has an ordered set of inputs, denoted
as $\text{In}(j)$, and an ordered set of outputs, denoted as
$\text{Out}(j)$.  Its resource consumption depends on the data rates
of the flows entering the component. We characterize this using a pair
of functions
$p_j,m_j:\mathbb{R}_{\ge 0}^{|\text{In}(j)|} \to \mathbb{R}_{\ge 0}$,
where $p_j$ is the CPU load and $m_j$ is the required memory size of
component $j$, depending on the data rate of the incoming flows. These
functions typically account for resource consumption due to processing
the input data flows as well as fixed, baseline consumption (even when
idle). Similarly, data rates of the outputs of the component are
determined as a function of the data rates on the inputs, specified as
$r_j:\mathbb{R}_{\ge 0}^{|\text{In}(j)|} \to \mathbb{R}_{\ge
  0}^{|\text{Out}(j)|}$.  Fig.~\ref{fig:component} shows examples for
functions $p_j,m_j,r_j$ that define the resource demands and output
data rates of an example component.

Each arc in $A_T$ connects an output of a component to an input of another component.

\emph{Source components} are special components in the template: they have no inputs, a single output with unspecified data rate, and zero resource consumption. In the example of Fig.~\ref{fig:template}, S is a source component, whereas the others are normal processing components.

\subsection{Overlays and sources}

A template specifies the types of components and the connections among them as well as their resource demands depending on the load.
A specific, deployable instantiation of a network service can be
derived by scaling its template, i.e., creating the necessary number
of instances for each component and linking the instances with each
other according to the requirements of the template. Depending on data
rates of the service flows and the locations in the network where the
flows start, different numbers of instances for each component might
be required. To model this, for each network service $T$, we define a
set of \emph{sources} $S(T)$. The members of $S(T)$ are tuples of the
form $(v,j,\lambda)$, where $v\in V$ is a node of the substrate
network, $j\in C_T$ is a source component, and
$\lambda\in\mathbb{R}_+$ is the corresponding data rate assigned to the output of this source component. Such a
tuple means that an instance of source component $j$ generates a flow from node $v$ with rate $\lambda$. Sources may represent populations of users, sensors, or any other component that can generate flows to be processed by the corresponding network service. Fig.~\ref{fig:sources} shows two example sources for the template of Fig.~\ref{fig:template}, located on different nodes of the substrate network.

\begin{figure}[tb]
\centering
\subfigure[\label{fig:template}A template]{\includegraphics[width=0.27\columnwidth]{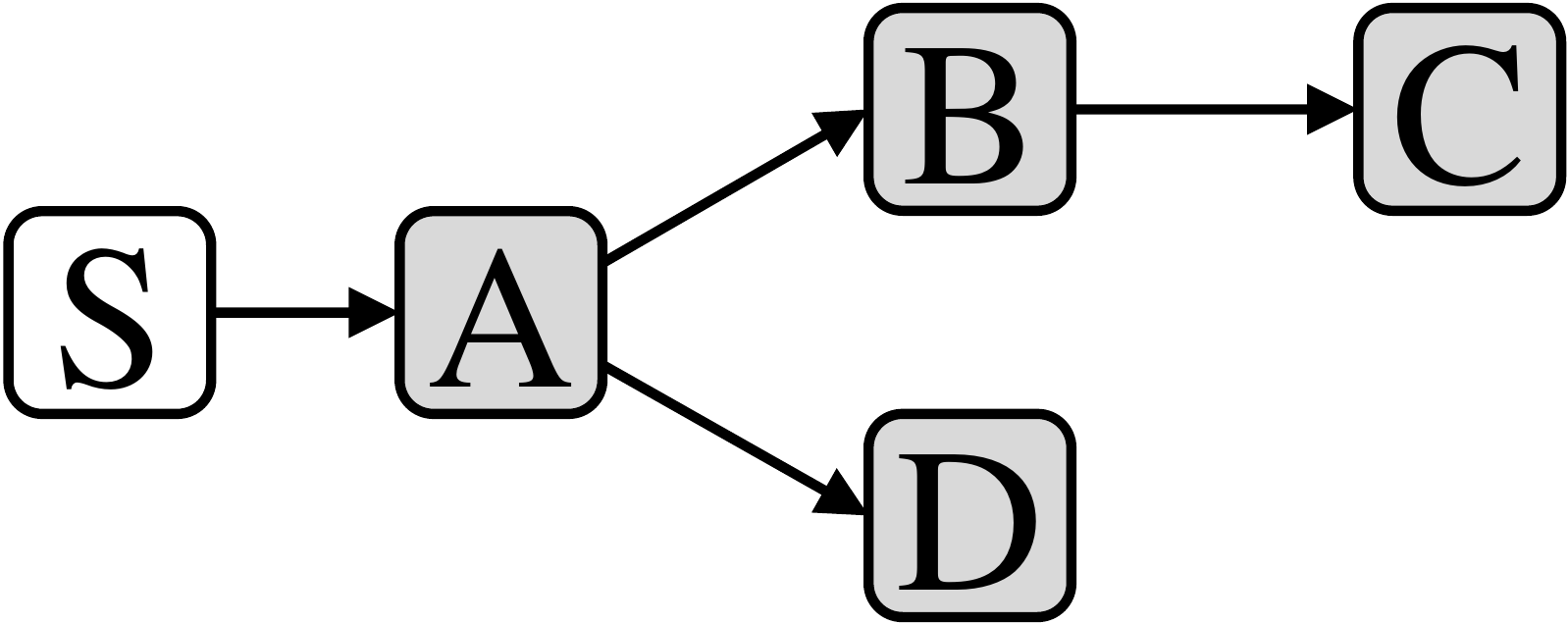}}
\hfil
\subfigure[\label{fig:component}A component]{\includegraphics[width=0.4\columnwidth]{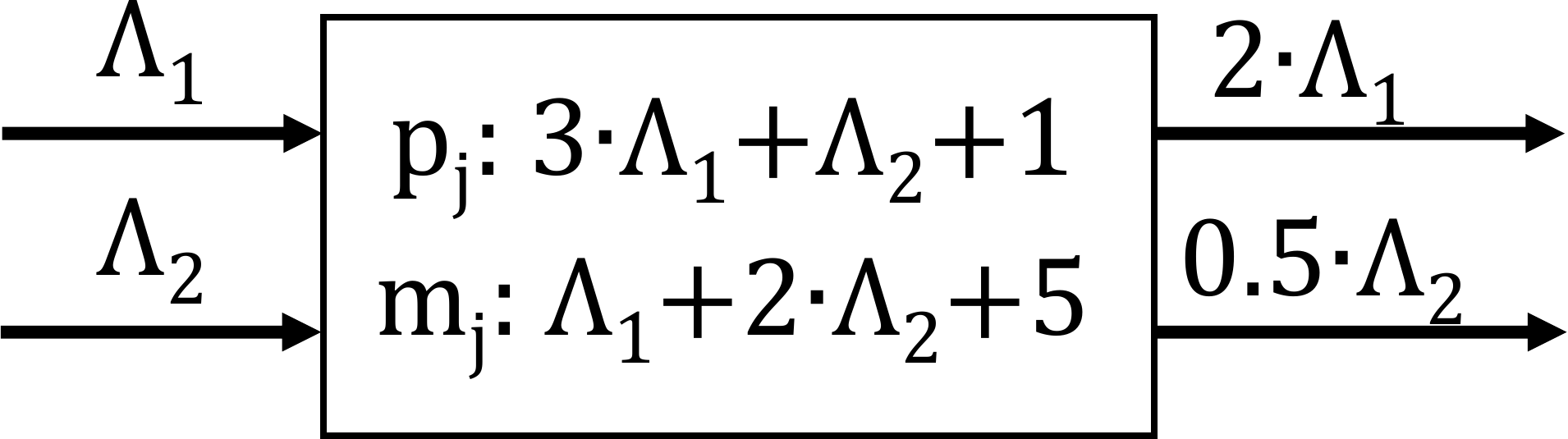}}

\subfigure[\label{fig:sources}Sources on a network]{\includegraphics[width=0.57\columnwidth]{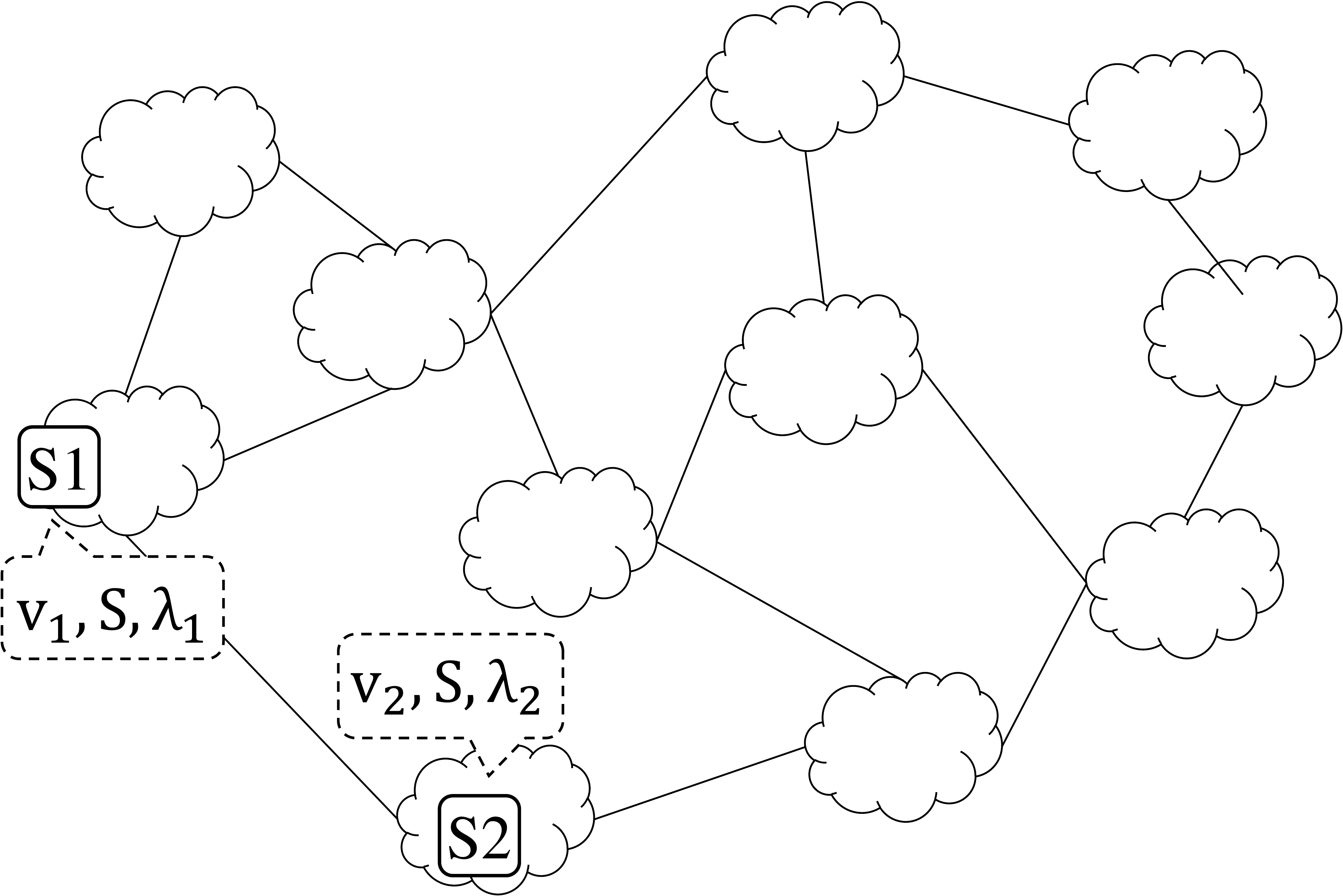}}

\subfigure[\label{fig:overlay}An overlay]{\includegraphics[width=0.3\columnwidth]{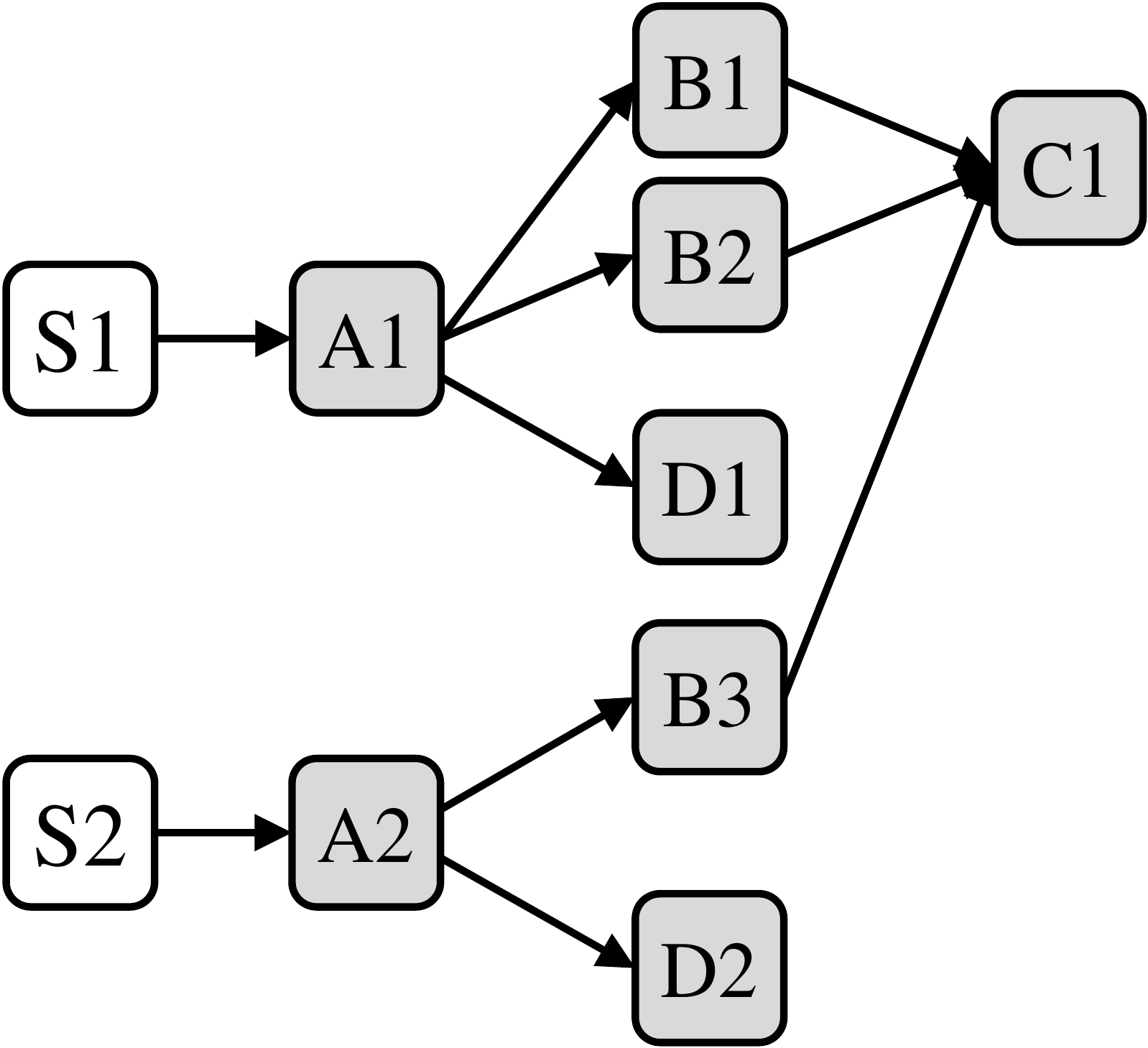}}
\hfil
\subfigure[\label{fig:mapping}Overlay embedded in the network]{\includegraphics[width=0.57\columnwidth]{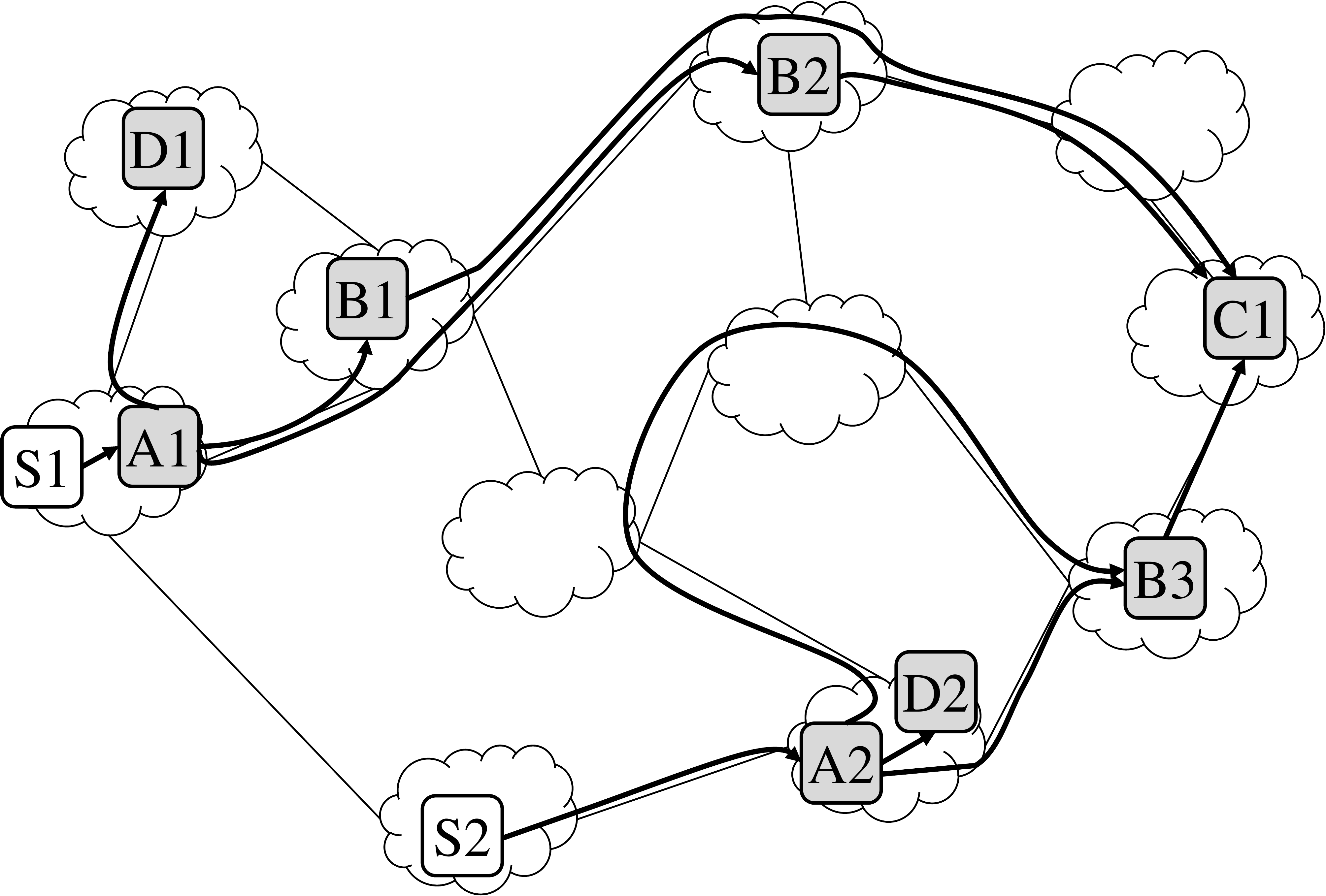}}
\caption{Some examples: (a) a template, (b) a component, (c) an overlay corresponding to the template, and (d) a mapping of the overlay into a substrate network. The links of the substrate network are bi-directional.}
\end{figure}

An \emph{overlay} is the outcome of scaling the template based on the associated sources. An overlay $\text{OL}$ stemming from template $T$ is described by a directed acyclic graph $G_\text{OL}(T){=}(I_\text{OL},E_\text{OL})$. Each component \emph{instance} $i\in I_\text{OL}$ corresponds to a component $c(i)\in C_T$ of the underlying template.
Each $i\in I_\text{OL}$ has the same characteristics (inputs, outputs, resource consumption characteristics) as $c(i)$. Moreover, if there is an edge from an output of an instance $i_1$ to an input of instance $i_2$ in the overlay, then there must be a corresponding arc from the corresponding output of $c(i_1)$ to the corresponding input of $c(i_2)$ in the template. This ensures that the edge structure of the overlay is in line with the structural requirements of the network service, represented by the arcs in the template.

To be able to create the required number of instances for each
component, we assume either that the components are stateless or that a state
management system is in place to handle state redistribution upon adding or removing instances. In this way, requests can be freely routed to any instance of a component. Alternatively, additional details can be added to the model, for example, to make sure that the flows belonging to a certain session are routed to the right instance of stateful components that have stored the corresponding state information. 

Fig.~\ref{fig:overlay} shows an example overlay corresponding to the template in Fig.~\ref{fig:template}. 
The naming of the instances follows the convention that the first letter identifies the corresponding component in the template, e.g., A1 is an instance of component A.
An overlay might include multiple instances of a specific template
component, e.g., B1, B2, and B3 all  are instances of component B.
An output of an instance can be connected to the input of multiple instances of the same component, like the output of A1 is connected to the inputs of B1 and B2. In a case like that, B1 and B2 share the data rate calculated for the connection between components A and B. Similarly, outputs of multiple instances in the overlay can be connected to the input of the same instance, like the input of C1 is connected to the output of B1, B2, and B3, in which case the input data rate for C1 is the sum of the output data rates of B1, B2, and B3.

\subsection{Mapping on the substrate network}

Each overlay $G_\text{OL}(T)$ must be mapped to the substrate network by a feasible mapping $P_T$. We define the mapping as a pair of functions $P_T=\left(P_T^{(I)},P_T^{(E)}\right)$.

$P_T^{(I)}:I_\text{OL}\to V$ maps each instance in the overlay to a
node in the substrate network.  We make the simplifying assumption
that two instances of the same component cannot be mapped to the same
node. The rationale behind this assumption is that in this case it
would be more efficient to replace the two instances by a single
instance and thus save the idle resource consumption of one
instance.\footnote{This simplification is mostly a technicality to
  simplify the problem write-up and could be extended if necessary.}

$P_T^{(E)}:E_\text{OL}\to {\cal F}$ maps each edge in the overlay to a flow in the substrate network; $\cal F$ is the set of possible flows in $G_\text{sub}$.
We assume the flows are splittable, i.e., can be routed over multiple paths between the corresponding endpoints in the substrate network.

The two functions must be compatible: if $e\in E_\text{OL}$ is an edge from an instance $i_1$ to an instance $i_2$, then $P_T^{(E)}(e)$ must be a flow with start node $P_T^{(I)}(i_1)$ and end node $P_T^{(I)}(i_2)$. Moreover, $P_T^{(I)}$ must map the instances of source components in accordance with the sources in $S(T)$, mapping an instance corresponding to source component $j$ to node $v$ if and only if $\exists (v,j,\lambda)\in S(T)$.

The binding of instances of source components to sources determines the outgoing data rate of these instances. As the overlay graphs are acyclic, the data rate $\lambda(e)$ on each further overlay edge $e$ can be determined based on the input data rates and the $r_j$ functions of the underlying components, considering the instances in a topological order. The data rates, in turn, determine the resource needs of the instances.


Fig.~\ref{fig:mapping} shows a possible mapping of the overlay of Fig.~\ref{fig:overlay} to an example substrate network, based on the pre-defined location of S1 and S2 in the network.
Note that it is possible to map two communicating instances to the same node, like A2 and D2 in the example. In this case, the edge between them can be realized inside the node, without using any links.
The flow between A2 and B3 is an example of a split flow that is routed over two different paths in the substrate network.

Note that Fig.~\ref{fig:mapping} shows only a single overlay mapped to
the substrate network for the sake of clarity. In general, JASPER can
embed several overlays corresponding to different network
services into a substrate network.




\subsection{Objectives}
\label{subsec:problem}

The \emph{system configuration} consists of the overlays and their mapping on
the substrate network.
A new system configuration can be computed by an appropriate algorithm for the template embedding problem.

A valid system configuration must respect all capacity constraints: for each node $v$, the total resource needs of the instances mapped to $v$ must be within its capacity concerning both CPU and memory, and for each link $l$, the sum of the flow values going through $l$ must be within its maximum data rate. However, it is also possible that some of those constraints are violated in a given system configuration: for example, a valid system configuration (i.e., one without any violations) may become invalid because the data rate of a source has increased, because of a temporary peak in resource needs, or a failure in the substrate network. Therefore, given a current system configuration $\sigma$, our primary objective is to find a new system configuration $\sigma'$, in which the \emph{number of constraint violations is minimal} (ideally, zero). For this, we assume that violating node CPU, memory, and link capacity constraints is equally undesirable.

There are a number of further, secondary objectives, which can be used as tie-breaker to choose from system configurations that have the same number of constraint violations:
\begin{itemize}
\item Total delay of all edges across all overlays
\item Number of instance addition/removal operations required to transition from $\sigma$ to $\sigma'$
\item Maximum  amounts of capacity constraint violations, for each resource type (CPU, memory, link capacity)
\item Total resource consumption of all instances across all overlays, for each resource type (CPU, memory, link capacity)
\end{itemize}
Higher values for these metrics result in higher costs for the system or in lower customer satisfaction, so our objective is to minimize these values. Therefore, our aim is to select a new system configuration $\sigma'$ from the set of system configurations with minimal number of constraint violations that is Pareto-optimal with respect to these secondary metrics.

\subsection{Problem formulation summary}

Our aim is to handle the scaling, placement, and routing for newly requested network services as well as already deployed network services.
Taking this into account, the Template Embedding problem can be summarized as follows:
\begin{itemize}
\item Inputs:
	\begin{itemize}
	\item Substrate network
	\item Template for each network service
	\item Location and data rate of the sources for each network service
	\item For the already deployed network services: overlay and its mapping onto the substrate network
	\end{itemize}
\item Outputs:
	\begin{itemize}
	\item For the newly requested network services: overlay and its mapping onto the substrate network
	\item For the already deployed network services: modified overlay and its modified mapping onto the substrate network
	\end{itemize}
\end{itemize}
Scaling is performed while creating the overlay from the template, while placement and routing are performed when the instances and edges of the overlay are mapped onto the substrate network.

A further important detail concerns the relationship between different network services. The creation of the overlay from the template and its mapping onto the substrate network are defined for each network service separately; however, they share the same substrate network. The objectives defined in Sec.~\ref{subsec:problem} relate to the whole network including all network services, aiming for a global optimum
and potentially resulting in trade-offs among the network services. 
A further connection among different network services may arise if they share the same component type. In this case, it is also possible that the corresponding overlay instances are realized by the same instance.

%% file: complexity.tex
\section{Complexity}
\label{sec:compl}

\begin{theorem}
\label{thm:npc}
For an instance of the Template Embedding problem as defined in Section~\ref{sec:problem}, deciding whether a solution with no violations exists is NP-complete in the strong sense\footnote{NP-complete in the strong sense means that the problem remains NP-complete even if the numbers appearing in it are constrained between polynomial bounds. Under the P$\ne$NP assumption, this precludes even the existence of a pseudo-polynomial algorithm -- i.e., an algorithm the runtime of which is polynomial if restricted to problem instances with polynomially bounded numbers.}.
\end{theorem}

\begin{proof}
It is clear that the problem is in NP: a possible witness for the positive answer is a solution -- i.e., a set of overlays and their embedding into the substrate network -- with 0 violations. The witness has polynomial size and can be verified in polynomial time wrt.\ to the input size.

To establish NP-hardness, we show a reduction from the Set Covering
problem (which is known to be NP-complete in the strong sense
\cite{karp1972reducibility}) to the Template Embedding
problem.
An input of the Set Covering problem consists of a finite set $U$, a finite family $\cal W$ of subsets of $U$ such that their union is $U$, and a number $k\in\mathbb{N}$. The aim is to decide whether there is a subset $\cal Z\subseteq W$ with cardinality at most $k$ such that the union of the sets in $\cal Z$ is still $U$.

From this instance of Set Covering, an instance of the Template Embedding problem is created as follows. The substrate network consists of nodes $V=\{s_1,\ldots,s_{|U|}\}\cup\{a_1,\ldots,a_{|{\cal W}|}\}\cup\{b\}$, where each $s_i$ represents an element of $U$ and each element $a_j$ represents an element of $\cal W$. There is a link from $s_i$ to $a_j$ if and only if the element of $U$ represented by $s_i$ is a member of the set represented by $a_j$. Furthermore, there is a link from each $a_j$ to $b$. The capacities of the nodes are as follows: $\text{cap}_{\text{cpu}}(s_i)=\text{cap}_{\text{mem}}(s_i)=0$ for each $i\in[1,|U|]$, $\text{cap}_{\text{cpu}}(a_j)=0$ and $\text{cap}_{\text{mem}}(a_j)=1$ for each $j\in[1,|{\cal W}|]$, and $\text{cap}_{\text{cpu}}(b)=1$ and $\text{cap}_{\text{mem}}(b)=0$. For each link, its maximum data rate is 1, its delay is 0.

There is a single template consisting of a source component $S$ and two further components $A$ and $B$, and two arcs $(S,A)$ and $(A,B)$. Component $A$ has one input and one output, its resource consumption as a function of the input data rate $\lambda$ is given by $p_A(\lambda)=0$ and $m_A(\lambda)=1$; its output data rate is given by $r_A(\lambda)=1$. Component $B$ has one input and no output, its resource consumption as a function of the input data rate $\lambda$ is given by
\begin{equation*}
p_B(\lambda)=
\begin{cases}
1, & \text{if }\lambda\le k, \\
2, & \text{otherwise,}
\end{cases}
\end{equation*}
and $m_B(\lambda)=0$. In each $s_i$, there is a source corresponding to an instance of $S$ with data rate $\lambda=1$.

\begin{figure}[tb]
\centering
\subfigure[An instance of Set Covering ($k=2$) and a solution (thick lines)]{\hspace*{0.05\columnwidth}\includegraphics[width=0.35\columnwidth]{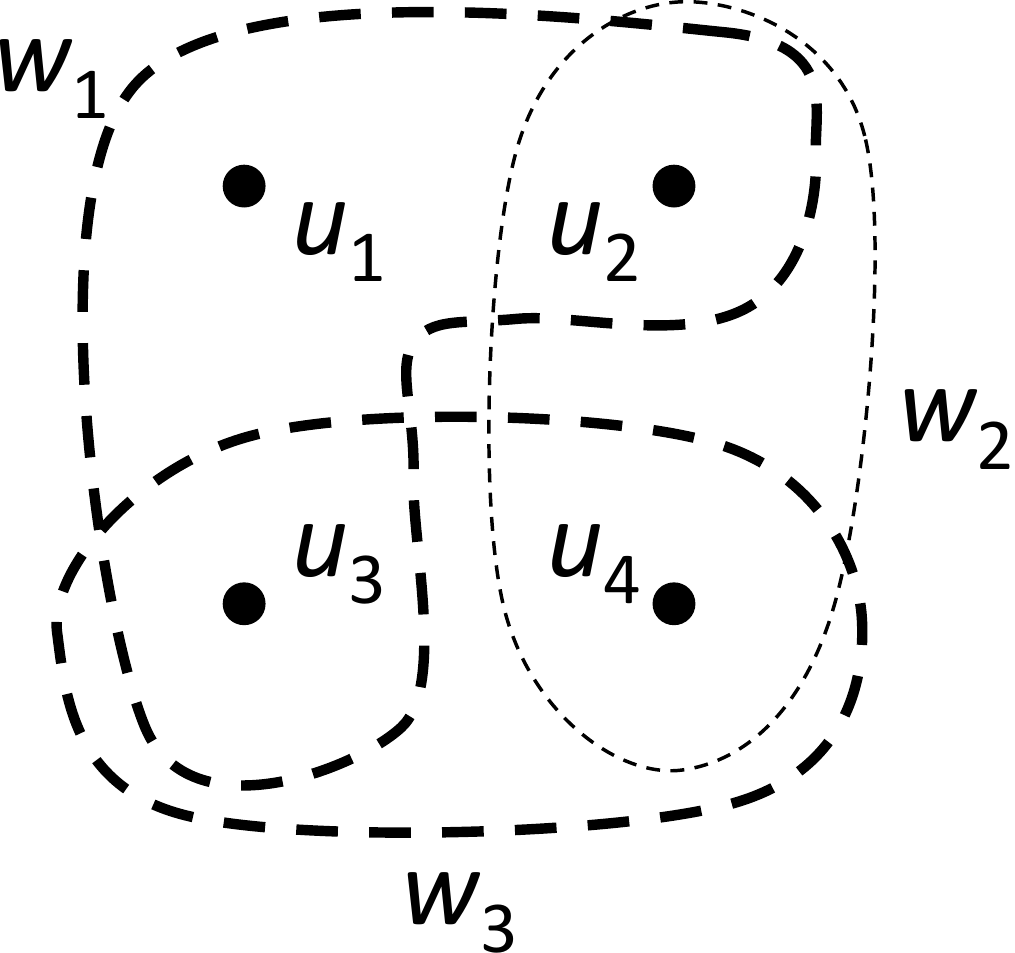}\hspace*{0.05\columnwidth}}
\hspace*{0.05\columnwidth}
\subfigure[The generated instance of Template Embedding and the corresponding solution]{\includegraphics[width=0.45\columnwidth]{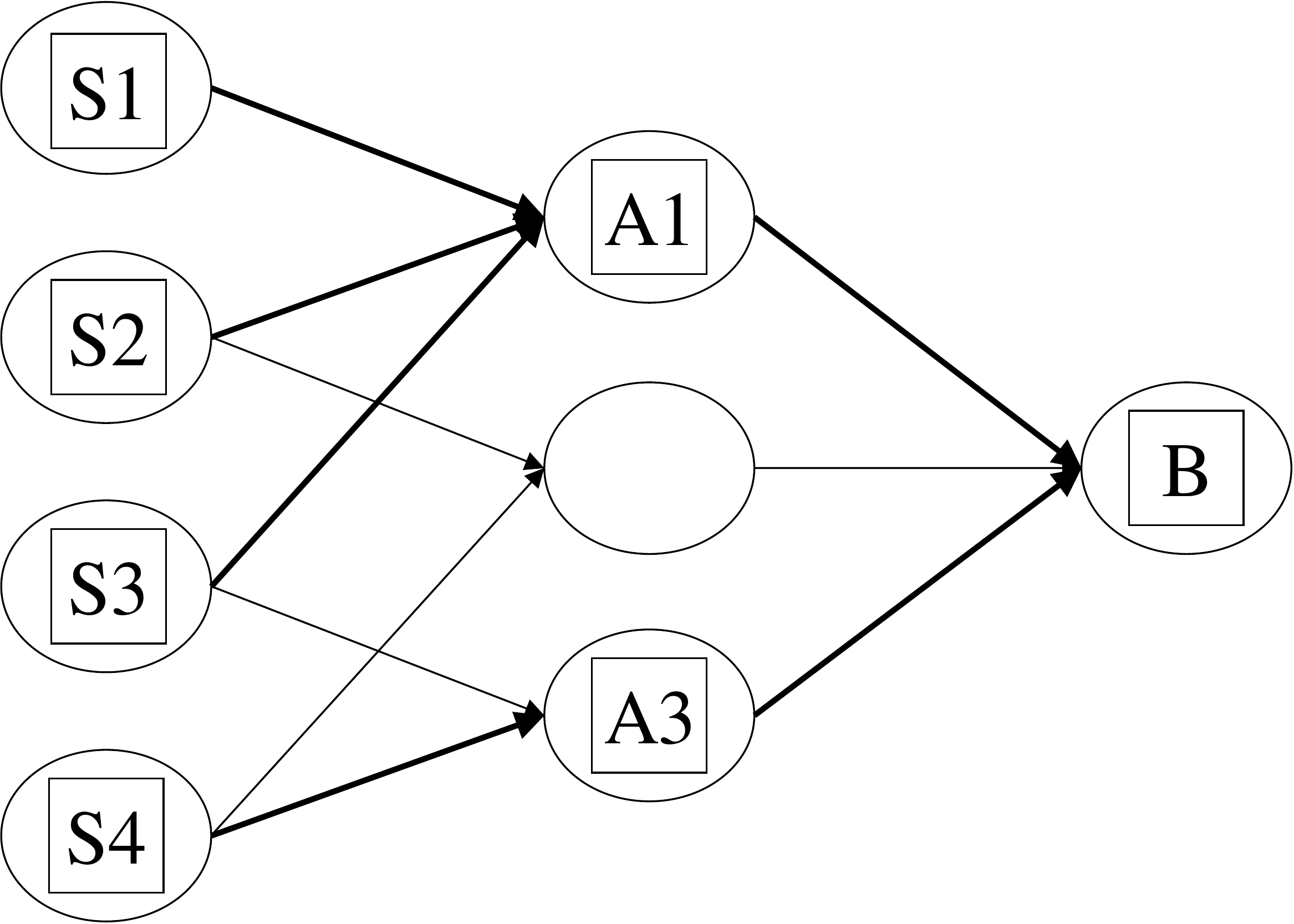}}
\caption{An example for the proof of Theorem \ref{thm:npc}}
\label{fig:proof_npc}
\end{figure}

Suppose first that the original instance of Set Covering is solvable, i.e., there is a subset $\cal Z\subseteq W$ with cardinality at most $k$ such that the union of the sets in $\cal Z$ is $U$. In this case, the generated instance of the Template Embedding problem can also be solved without any violations, as follows (see Fig.\ \ref{fig:proof_npc} for an example). Each $s_i$ must of course host an instance of $S$. In each $a_i$ corresponding to an element of $\cal Z$, an instance of $A$ is created. Since the union of the sets in $\cal Z$ is $U$, each $s_i$ has an outgoing link to at least one $a_j$ hosting an instance of $A$, which can be selected as the target of the traffic leaving the source in $s_i$ through the link $(s_i,a_j)$. Further, a single instance of $B$ is created in node $b$ and each instance of $A$ is connected to $B$ through the $(a_j,b)$ link. Since the number of instances of $A$ is at most $k$, each emitting traffic with data rate 1, the CPU requirement of the instance of $B$ is 1, so that it fits on $b$, and hence we obtained a solution to the Template Embedding problem with no violation.

Now assume that the generated instance of the Template Embedding problem is solvable without violations. Then, we can construct a solution of the  original instance of Set Covering, as we show next. In a solution of the generated instance of the Template Embedding problem, each $s_i$ must host an instance of $S$ and there is no other instance of $S$. Instances of $A$ can only be hosted by $a_j$ nodes because of the memory requirement, and an instance of $B$ can only be hosted in $b$ because of the CPU requirement. We define $\cal Z$ to contain those elements of $\cal W$ for which the corresponding node $a_j$ hosts an instance of $A$. Since each source generates traffic that must be consumed by an instance of $A$ and there is a path (actually, a link) from $s_i$ to $a_j$ only if the set corresponding to $a_j$ contains the element corresponding to $s_i$, it follows that the sets in $\cal Z$ cover all elements of $U$. Moreover, since the instance of $B$ must fit on $b$ and each instance of $A$ generates traffic with data rate 1, it follows that the number of instances of $A$ is at most $k$ and hence $|{\cal Z}|\le k$, thus $\cal Z$ is a solution of the original Set Covering problem.

Since all numbers in the generated instance of the Template Embedding problem are constants, this reduction shows that the Template Embedding problem is indeed NP-hard in the strong sense.
\end{proof}

As a consequence, we can neither expect a polynomial or even pseudo-polynomial algorithm for solving the problem exactly nor a fully polynomial-time approximation scheme, under standard assumptions of complexity theory.


%% file: integer.tex
\section{Mixed integer programming approach}
\label{sec:integer}

In this section, we provide a mixed integer programming (MIP) formulation of the problem. On one hand, this serves as a further formalization of the problem; on the other hand, under suitable assumptions (to be detailed in Section~\ref{sec:solve_mip}) an appropriate solver can be used to solve the mixed integer program, yielding an algorithm for the problem.

Based on the assumption that two instances of the same component cannot be mapped to a node, instances can be identified by the corresponding component and the hosting node. This is the basis for our choice of variables, which are explained in more detail in Table~\ref{tab:vars}.

\begin{table}[tb]
\caption{Variables}
\label{tab:vars}
\begin{tabular}{p{8mm}p{8mm}p{60mm}}
\toprule
Name & Domain & Definition \\
\midrule
$x_{j,v}$ & $\{0,1\}$ & 1 iff an instance of component $j{\in} {\cal C}$ is mapped to node $v{\in} V$ \\
$y_{a,v,v'}$ & $\mathbb{R}_{\ge 0}$ & If $a{\in} A_T$ is an arc from an output of $j{\in} C_T$ to an input of $j'{\in} C_T$, an instance of $j$ is mapped on $v{\in} V$, and an instance of $j'$ is mapped on $v'{\in} V$, then $y_{a,v,v'}$ is the data rate of the corresponding flow from $v$ to $v'$; otherwise it is 0 \\
$z_{a,v,v',l}$ & $\mathbb{R}_{\ge 0}$ & If $a{\in} A_T$ is an arc from an output of $j{\in} C_T$ to an input of $j'{\in} C_T$, an instance of $j$ is mapped on $v{\in} V$, and an instance of $j'$ is mapped on $v'{\in} V$, then $z_{a,v,v',l}$ is the data rate of the corresponding flow from $v$ to $v'$ that goes through link $l{\in} L$; otherwise it is 0 \\
$\Lambda_{j,v}$ & $\mathbb{R}_{\ge 0}^{|\text{In}(j)|}$ & Vector of data rates on the inputs of the instance of component $j{\in} C_T$ on node $v{\in} V$, or an all-zero vector if no such instance is mapped on $v$ \\
$\Lambda'_{j,v}$ & $\mathbb{R}_{\ge 0}^{|\text{Out}(j)|}$ & Vector of data rates on the outputs of the instance of component $j{\in} C_T$ on node $v{\in} V$, or an all-zero vector if no such instance is mapped on $v$ \\
$\varrho_{j,v}$ & $\mathbb{R}_{\ge 0}$ & CPU requirement of the instance of component $j{\in} C_T$ on node $v{\in} V$, or zero if no such instance is mapped on $v$ \\
$\mu_{j,v}$ & $\mathbb{R}_{\ge 0}$ & Memory requirement of the instance of component $j{\in} C_T$ on node $v{\in} V$, or zero if no such instance is mapped on $v$ \\
$\omega_{v,\text{cpu}}$ & $\{0,1\}$ & 1 iff the CPU capacity of node $v{\in} V$ is exceeded \\
$\omega_{v,\text{mem}}$ & $\{0,1\}$ & 1 iff the memory capacity of node $v{\in} V$ is exceeded \\
$\omega_l$ & $\{0,1\}$ & 1 iff the maximum data rate of link $l{\in} L$ is exceeded \\
$ \psi_\text{cpu} $ & $\mathbb{R}_{\ge 0}$ & Maximum CPU over-subscription over all nodes \\
$ \psi_\text{mem} $ & $\mathbb{R}_{\ge 0}$ & Maximum memory over-subscription over all nodes \\
$ \psi_\text{dr} $ & $\mathbb{R}_{\ge 0}$ & Maximum capacity over-subscription over all links\\
$\zeta_{a,v,v',l}$ & $\{0,1\}$ & 1 iff $z_{a,v,v',l}>0$ \\
$\delta_{j,v}$ & $\{0,1\}$ & 1 iff $x_{j,v}\ne x^*_{j,v}$ \\
\bottomrule
\end{tabular}
\end{table}


We use the following notations for formalizing the constraints and objectives. ${\cal C}{=}\bigcup_{T\in {\cal T}} C_T$ denotes the set of all components, ${\cal A}{=}\bigcup_{T\in {\cal T}} A_T$ the set of all arcs, and ${\cal S}{=}\bigcup_{T\in {\cal T}} S(T)$ the set of all sources across all network services that we want to map to the network. $M$, $M_1$, and $M_2$ denote sufficiently large constants. $(\Lambda_{j,v})_k$ denotes the $k$th component of the vector $\Lambda_{j,v}$. $\underbar{0}$ denotes a zero vector of appropriate length. 

Information about existing instances should also be taken into account during the decision process. For this, we define $x^*_{j,v} (\forall j \in {\cal C}, v \in V)$ as a constant given as part of the problem input. If there is a previously mapped instance of component $j$ on node $v$ in the network, $x^*_{j,v}$ is 1, otherwise it is 0.

\subsection{Constraints}

Here we define the sets of constraints that enforce the required
properties of the template embedding process.



\subsubsection{Mapping consistency rules}

{\footnotesize
\begin{align}
\forall (v,j,\lambda)\in {\cal S}: && x_{j,v} & = 1 \\
\forall (v,j,\lambda)\in {\cal S}: && \Lambda'_{j,v} & = \lambda \\
\forall j\in {\cal C}, \forall v\in V, k \in [1,|\text{In}(j)|]: && (\Lambda_{j,v})_k & \le  M\cdot x_{j,v} \\
\forall j\in {\cal C}, \forall v\in V, k \in [1,|\text{Out}(j)|]: && (\Lambda'_{j,v})_k & \le  M\cdot x_{j,v} \\
\forall j\in {\cal C}, \forall v\in V: && x_{j,v}-x^*_{j,v} & \le \delta_{j,v} \\
\forall j\in {\cal C}, \forall v\in V: && x^*_{j,v}-x_{j,v} & \le \delta_{j,v}
\end{align}
}

Constraints (1) and (2) enforce that the placement respectively the output data rate of source component instances are in line with the tuples specified in $\cal S$. Constraint (3) guarantees the consistency between the variables $\Lambda_{j,v}$ and $x_{j,v}$: if $\Lambda_{j,v}$ has a positive component, then $x_{j,v}$ must be 1, i.e., only an existing component instance can process the incoming flow. Constraint (3) is analogous for the outgoing flows, represented by the $\Lambda'_{j,v}$ variables. Constraints (5) and (6) together ensure that $\delta_{j,v}=1$ if and only if $x_{j,v}\ne x^*_{j,v}$.

\subsubsection{Flow and data rate rules}
{\footnotesize
\begin{multline}
\forall j\in {\cal C}, \text{$j$ not a source component}, \forall v\in V:\\
\Lambda'_{j,v} = r_j(\Lambda_{j,v})-(1-x_{j,v})\cdot r_j(\underbar{0})
\end{multline}
\begin{multline}
\forall j\in {\cal C}, \forall v\in V, k \in [1,|\text{In}(j)|]:\\
(\Lambda_{j,v})_k=\sum_{a\text{ ends in input }k\text{ of }j, v'\in V} y_{a,v',v}
\end{multline}
\begin{multline}
\forall j\in {\cal C}, \forall v\in V, k \in [1,|\text{Out}(j)|]:\\
(\Lambda'_{j,v})_k=\sum_{a\text{ starts in output }k\text{ of }j, v'\in V} y_{a,v,v'}
\end{multline}
\begin{multline}
\forall a\in{\cal A},\forall v,v_1,v_2\in V:\\
\sum_{vv'\in L}z_{a,v_1,v_2,vv'}-\sum_{v'v\in L}z_{a,v_1,v_2,v'v} =\\ 
=\begin{cases}
0 & \text{if }v\ne v_1\text{ and }v\ne v_2 \\
y_{a,v_1,v_2} & \text{if }v=v_1\text{ and }v_1\ne v_2 \\
0 & \text{if }v=v_1=v_2
\end{cases}
\end{multline}
\begin{align}
\forall a\in{\cal A}, \forall v,v'\in V, \forall l\in L: && z_{a,v,v',l} & \le M\cdot \zeta_{a,v,v',l}
\end{align}
}

Constraint (7) computes the data rate on the outputs of a processing component instance based on the data rates on its inputs and the $r_j$ function of the underlying component. The constraint is formulated in such a way that for $x_{j,v}=1$, $\Lambda'_{j,v} = r_j(\Lambda_{j,v})$, whereas for $x_{j,v}=0$ (in which case also $\Lambda_{j,v}=0$ because of Constraint (3)), also $\Lambda'_{j,v}=0$ so that there is no contradiction with Constraint (4). Constraint (8) computes the data rate on the inputs of a component instance as the sum of the data rates on the links ending in that input. Similarly, Constraint (9) ensures that the data rate on the outputs of a component instance is distributed on the links starting in that output. Constraint (10) is the flow conservation rule, also ensuring the right data rate of each flow, thus connecting the $z_{a,v,v',l}$ variables (flow values on individual links) and the $y_{a,v,v'}$ variables (flow data rate). Constraint (11) sets the $\zeta_{a,v,v',l}$ variables (on the basis of the $z_{a,v,v',l}$ variables), so that they can be used later on in the objective function (Section~\ref{subsec:objective}).

\subsubsection{Calculation of resource consumption}

{\footnotesize
\begin{align}
\forall j\in {\cal C}, \forall v\in V: && \varrho_{j,v} & = p_j(\Lambda_{j,v})-(1-x_{j,v})\cdot p_j(\underbar{0}) \\
\forall j\in {\cal C}, \forall v\in V: && \mu_{j,v} & = m_j(\Lambda_{j,v})-(1-x_{j,v})\cdot m_j(\underbar{0})
\end{align}
}

Constraints (12) and (13) calculate CPU respectively memory
consumption of each component instance based on the $p_j$ and $m_j$
functions of the underlying component\footnote{Adding more resource types would be
  reflected by adding corresponding constraints here.}. The logic here is analogous to
that of Constraint (7).

\subsubsection{Capacity constraints}
{\footnotesize
\begin{align}
\forall v\in V: & \quad \sum_{j\in {\cal C}}\varrho_{j,v}  \le \text{cap}_{\text{cpu}}(v)+M\cdot\omega_{v,\text{cpu}} \\
\forall v\in V: & \quad \sum_{j\in {\cal C}}\varrho_{j,v} - \text{cap}_{\text{cpu}}(v) \leq \psi_\text{cpu} \\
\forall v\in V: & \quad \sum_{j\in {\cal C}}\mu_{j,v}  \le \text{cap}_{\text{mem}}(v)+M\cdot\omega_{v,\text{mem}} \\
\forall v\in V: & \quad \sum_{j\in {\cal C}}\mu_{j,v} - \text{cap}_{\text{mem}}(v) \leq \psi_\text{mem} \\
\forall l\in L: & \quad \sum_{a\in {\cal A};v,v'\in V}z_{a,v,v',l} \le b(l)+M\cdot\omega_l \\
\forall l\in L: & \quad \sum_{a\in {\cal A};v,v'\in V}z_{a,v,v',l} - b(l) \leq \psi_\text{dr}
\end{align}
}

The aim of these constraints is to set the $\omega$ and $\psi$
variables (based on the already defined $\varrho$, $\mu$ and $z$
variables), which will be used in the objective function (Section~\ref{subsec:objective}). Constraint (14) ensures that $\omega_{v,\text{cpu}}$ will be 1 if the CPU capacity of node $v$ is overloaded, while Constraint (15) ensures that $\psi_\text{cpu}$ will be at least as high as the amount of CPU overload of any node (the appearance of $\psi_\text{cpu}$ in the objective function will guarantee that it will be exactly the maximum amount of CPU overload and not higher than that). Constraints (16), (17) do the same for memory overloads and Constraints (18), (19) do the same for the overload of link capacity.

\subsubsection{Interplay of the constraints}

To illustrate the interplay of the constraints, we assume that we need to
optimize the embedding shown in 
Fig.~\ref{fig:mapping}. Constraints (1) and (2) ensure that instances of the
source component, i.e., S1 and S2, are embedded and their output
data rates are set correctly. Constraint (9) ensures that these data rates are
then handed out as flows that can only end up in instances of A.
These flows are mapped to network links and instances of A are assigned input
data rates using Constraints
(10) and (8), respectively.
That being set, Constraint (3) marks the instances A1 and A2
as embedded, and Constraint (7) sets their output data rates using the respective $r_j$ function. In a similar 
way, the rest of the components are instantiated and embedded in the network. 

Constraints (5) and (6) ensure that the $\delta_{j,v}$ variables are set correctly.  Constraints (12) and (13) compute the resource
consumption of each instance based on the input data rates and the 
corresponding $p_j$ and $m_j$ functions. Constraints (14)--(19) make sure that
over-subscription of node and link capacities are captured correctly, and 
collect the maximum value of over-subscription for each resource type. This 
maximum value is used in the objective function described in Section~\ref{subsec:objective}, which drives the decisions based on the constraints. 

\subsection{Optimization objective}
\label{subsec:objective}

We formalize the optimization objective based on the goals defined in 
Section~\ref{subsec:problem} as follows:

{\footnotesize
\begin{multline}
\text{minimize} \quad M_1\cdot\Big( \sum_{v\in V}(\omega_{v,\text{cpu}} + \omega_{v,\text{mem}}) + \sum_{l\in L}\omega_{l}\Big)+ \\
+ M_2\cdot\Big(\sum_{\substack{a\in {\cal A} \\ v,v'\in V \\ l\in L}}(d(l)\cdot\zeta_{a,v,v',l})+ \sum_{\substack{j\in {\cal C} \\ v\in V}}\delta_{j,v}\Big)+ \\
+ \psi_\text{cpu} + \psi_\text{mem} + \psi_\text{dr} + \sum_{\substack{j\in {\cal C} \\ v\in V}}(\varrho_{j,v} + \mu_{j,v}) + \sum_{\substack{a\in {\cal A} \\ v,v'\in V \\ l\in L}}z_{a,v,v',l}
\end{multline}
}

By assigning sufficiently large values to $M_1$ and $M_2$,
we can
achieve the following goals with the given priorities (1 being the highest 
priority):
\begin{enumerate}
     \item Number of capacity constraint violations over all nodes and links is minimized.
     \item Template arcs are mapped to network paths in such a way that their total latency is minimized. Moreover, the
     number of instances that need to be started/stopped is minimized.
     \item The maximum value for capacity constraint violations over all nodes and links is minimized. Also, overlay instances and the edges among them are created in a way that their resource consumption is minimized. 
\end{enumerate} 

The objective function is in line with the objectives defined in Section~\ref{subsec:problem}. The primary objective is to minimize the number of constraint violations; a sufficiently large $M_1$ ensures that a decrease in the first term of the objective function has larger impact than any change in the other terms. Moreover, the resulting solution $\sigma'$ will be Pareto-optimal with respect to the other, secondary metrics: otherwise, there would be another solution $\sigma''$ that is as good as $\sigma'$ according to each secondary metric and strictly better than $\sigma'$ in at least one secondary metric, but then, $\sigma''$ would lead to a lower overall value of the objective function.

This mixed integer program can be used for initial embedding of service templates
as well as for optimizing existing embeddings. However, for the initial embedding of newly requested network services,
the term $\sum_{j\in {\cal C}, v\in V}\delta_{j,v}$ 
should be removed from the objective function because it would introduce an unwanted bias towards embeddings with fewer instances, although it is possible that having more instances can decrease the overall cost of the solution.

\subsection{Solving the mixed integer program}
\label{sec:solve_mip}

All our constraints are linear equations and linear inequalities, and also the objective function is linear. Hence, if the functions $p_j$, $m_j$, and $r_j$ are linear for all $j\in{\cal C}$, then we obtain a mixed-integer linear program (MILP), which can be solved by appropriate solvers. For non-linear functions, a piecewise linear approximation may make it possible to use MILP solvers to obtain good (although not necessarily optimal) solutions.


%% file: heuristic.tex
\section{Heuristic approach}
\label{sec:heur}

Now we present a heuristic algorithm that is not guaranteed to find an
optimal solution but is much faster than the mixed integer programming
approach. Moreover, it has the advantage that it does not require
the functions $p_j$, $m_j$, and $r_j$ to be linear.


The heuristic constructs the new solution from the existing one by
means of a series of small local changes.\footnote{Also the placement of a new service is done with a series of small local changes, creating component instances one by one.}
While doing so, it has to be ensured that (i) the instantiation of source components is in line with the given data sources, (ii) the data flows produced by each instance are routed to appropriate instances, and (iii) capacity constraints are satisfied as much as possible. This can be achieved by iterating through the instances of each overlay once in a topological order, possibly creating new instances on the fly if necessary. Note that this may indeed be necessary, for example, if a new data source appeared or the output data rate of a data source increased. In each step, the algorithm aims at economical use of resources, e.g., by only creating new instances if necessary, deleting unneeded instances, or preferring short paths.

The heuristic is shown in Algorithm~\ref{alg:heur_main}. It starts by checking that each service has a corresponding overlay and each overlay corresponds to a service (lines 1--5). If a new service has been started or an existing service has been stopped since the last invocation of the algorithm, the corresponding overlay is created or removed at this point.

Next, the mapping of the sources and source components is checked and updated if necessary (lines 6--11): if a new source emerged, an instance of the corresponding source component is created; if the data rate of a source changed, then the output data rate of the corresponding source component instance is updated; if a source disappeared, then the corresponding source component instance is removed.

Finally, to propagate the changes of the sources to the processing instances, we need to iterate over all instances and ensure that the new output data rates, which are determined by the new input data rates, are discharged correctly by outgoing flows (lines 12--24). For this purpose, it is important to consider the instances in a topological order (according to the overlay) so that when an instance is dealt with, its incoming flows have already been updated. If a change in the outgoing flows is necessary, then the \textsc{increase} or \textsc{decrease} procedures are called.

\begin{algorithm}[htb]
\caption{Main procedure of the heuristic algorithm}
\label{alg:heur_main}
\begin{algorithmic}[1]
\small
\If{$\exists G_{\text{OL}}(T)$ with $T\not\in {\cal T}$}
	\State remove $G_{\text{OL}}(T)$
\EndIf
\ForAll{$T\in {\cal T}$}
	\If{$\nexists G_{\text{OL}}(T)$}
		\State create empty overlay $G_{\text{OL}}(T)$
	\EndIf
	\ForAll{$(v,j,\lambda)\in S(T)$}
		\If{$\nexists i\in I_{\text{OL}}$ with $c(i)=j$ and $P_T^{(I)}(i)=v$}
			\State create $i\in I_{\text{OL}}$ with $c(i)=j$ and $P_T^{(I)}(i)=v$
		\EndIf
		\State set output data rate of $i$ to $\lambda$
	\EndFor
	\If{$\exists i\in I_{\text{OL}}$, where $c(i)$ is a source component but $\nexists (P_T^{(I)}(i),c(i),\lambda)\in S(T)$ for any $\lambda$}
		\State remove $i$
	\EndIf
	\ForAll{$i\in I_{\text{OL}}$ in topological order}
		\If{all input data rates of $i$ are 0}
			\State remove $i$ and go to next iteration
		\EndIf
		\State compute output data rates of $i$
		\ForAll{output $k$ of $i$}
			\State $\Phi$: set of flows currently leaving output $k$
			\State $\lambda$: sum of the data rates of the flows in $\Phi$
			\State $\lambda'$: new data rate on output $k$
			\If{$\lambda'<\lambda$}
				\State $\cal E$: set of edges leaving output $k$
				\State \Call{decrease}{$\cal E$,$\lambda-\lambda'$}
			\ElsIf{$\lambda'>\lambda$}
				\State \Call{increase}{$i$,$k$,$\Phi$,$\lambda'-\lambda$}
			\EndIf
		\EndFor
	\EndFor
\EndFor
\end{algorithmic}
\end{algorithm}

The auxiliary subroutines are detailed in Algorithm~\ref{alg:aux}. \textsc{decrease} first removes as many edges as possible (lines 3--6); when a further decrease is necessary but no more edges can be removed, it reduces the next flow on each link by the same factor to achieve the required reduction (lines 7--9). \textsc{increase} first checks if new instances need to be created to be consistent with the template (lines 12--16), then tries to increase the existing flows (lines 17--19). If this is not sufficient to achieve the necessary increase, it creates further instances and flows (lines 20--23).

\begin{algorithm}[htb]
\caption{Auxiliary methods of the heuristic}
\label{alg:aux}
\begin{algorithmic}[1]
\small
\State\Comment{Decrease the flows on the edges in $\cal E$ by $\Delta\lambda$ in total}
\Procedure{decrease}{$\cal E$,$\Delta\lambda$}
	\State sort $\cal E$ in non-decreasing order of flow data rate
	\ForAll{$e\in {\cal E}$ while flow data rate $\lambda(e)\le\Delta\lambda$}
		\State $\Delta\lambda:=\Delta\lambda-\lambda(e)$
		\State remove $e$
	\EndFor
	\If{$\Delta\lambda>0$}
		\State let $e$ be the next edge
		\State reduce flow of $e$ by a factor of $(\lambda(e)-\Delta\lambda)/\lambda(e)$
	\EndIf
\EndProcedure
\State\Comment{Increase the flows in $\Phi$ leaving output $k$ of instance $i$ by $\Delta\lambda$ in total}
\Procedure{increase}{$i$,$k$,$\Phi$,$\Delta\lambda$}
	\ForAll{arc $(c(i),j)$ leaving output $k$ of $c(i)$}
		\If{$\nexists i'\in I_{\text{OL}}$ with $c(i')=j$ and $ii'\in E_{OL}$}
			\State $\varphi:=\;$\Call{createInstanceAndFlow}{$j$, $i$, $\Delta\lambda$}
			\State $\Delta\lambda:=\Delta\lambda-(\text{data rate of } \varphi)$
			\State $\Phi:=\Phi\cup\{\varphi\}$
		\EndIf
	\EndFor
	\ForAll{$\varphi\in\Phi$}
		\State $d:=\;$\Call{incrFlow}{$\varphi$,$\Delta\lambda$}
		\State $\Delta\lambda:= \Delta\lambda-d$
	\EndFor
	\While{$\Delta\lambda>0$}
		\State $(c(i),j)$: random arc leaving output $k$ of $c(i)$
		\State $\varphi:=\;$\Call{createInstanceAndFlow}{$j$, $i$, $\Delta\lambda$}
		\State $\Delta\lambda:=\Delta\lambda-(\text{data rate of }\varphi)$
	\EndWhile
\EndProcedure
\State\Comment{Create an instance of component $j$ with flow from instance $i$ of high data rate (capped at cutoff)}
\Procedure{createInstanceAndFlow}{$j$,$i$,cutoff}
	\ForAll{$v\in V$}
		\State create temporary instance $i'$ of $j$ on $v$
		\State $\varphi$: flow of data rate 0 from $i$ to $i'$
		\State \Call{incrFlow}{$\varphi$,cutoff}
		\State remove $i'$ and $\varphi$
	\EndFor
	\State create instance of $j$ on node resulting in best flow
\EndProcedure
\State\Comment{Increase flow data rate by at most $d$}
\Procedure{incrFlow}{$\varphi$,$d$}
	\State $v:=\;$start node of $\varphi$
	\State $v':=\;$end node of $\varphi$
	\State $\beta_1:=\;$maximum flow based on $cap_{CPU}(v')$
	\State $\beta_2:=\;$maximum flow based on $cap_{mem}(v')$
	\State $d:=\min(d,\beta_1,\beta_2)$
	\State $P$: $v\leadsto v'$ path of high bandwidth ($b$) and low latency
	\State increase $\varphi$ by $\min(b,d)$ along $P$
\EndProcedure
\end{algorithmic}
\end{algorithm}

In the \textsc{createInstanceAndFlow} procedure (called by \textsc{increase} to create a new instance of a component together with a flow from an existing instance), all nodes of the substrate network are temporarily tried for hosting the new instance. The candidate that leads to the best flow is selected (lines 26--31). Finally, the \textsc{incrFlow} procedure (called by both \textsc{increase} and \textsc{createInstanceAndFlow}) increases the data rate of a flow along a new path (lines 34--40).

As can be seen, we avoid computing maximum flows. This is because the
running time of the best known algorithms for this purpose are worse
than quadratic with respect to the size of the graph
\cite{hochbaum2008pseudoflow}.  Since these subroutines are run many
times, the high time complexity would be problematic for large
substrate networks. Instead, each run of \textsc{incrFlow} increases a
flow only along one new path. For finding the path, a modified
best-first-search \cite{korf1993linear} is used, which runs in linear time. It should be noted that split flows can still be created if \textsc{incrFlow} is run multiple times for a flow.

When improving a flow and when selecting from multiple possible flows,
the \textsc{incrFlow} and \textsc{createInstanceAndFlow} routines must
strike a balance between flow data rate and the increase in overall
delay of the solution. Our strategy for comparing two possible flows
is to first compare their data rates and compare their latencies only
if there is a tie. This strategy is used in line 31 to select the best
flow. The rationale is that selecting flows with high data rate leads
to a small number of instances to be created. However, we also employ
a cutoff mechanism: flow data rates above the cutoff (the increase in
data rate that we want to achieve) do not add more value and are hence
regarded to be equal to the cutoff value. This increases the
likelihood of a tie, so that the tie-breaking method of preferring
lower latencies is also important. An analogous strategy is used in
line 39 to compare paths: the primary criterion is to prefer paths
with higher bandwidth -- up to the given cutoff $d$ -- and, in case of
a tie, to prefer paths with lower latency. For finding the best path,
a modified best-first-search is used, in which the nodes to be visited are stored in a priority queue, where priority is defined in accordance with the comparison relation described above.





%% file: evaluation.tex
\section{Evaluation}
\label{sec:eval}

We implemented the presented algorithms in the form of a C++ program. For solving the MILP, Gurobi Optimizer 7.0.1\footnote{\url{http://www.gurobi.com/}} was used. For substrate networks, we used benchmarks for the Virtual Network Mapping Problem\footnote{\url{https://www.ac.tuwien.ac.at/files/resources/instances/vnmp}} from Inf\"uhr and Raidl \cite{infuhr2013solving}. As service templates, we used examples from IETF's Service Function Chaining Use Cases \cite{draft-liu-sfc-use-cases-08}.

\subsection{An example}

\begin{figure}[tb]
\centering
\includegraphics[height=55mm]{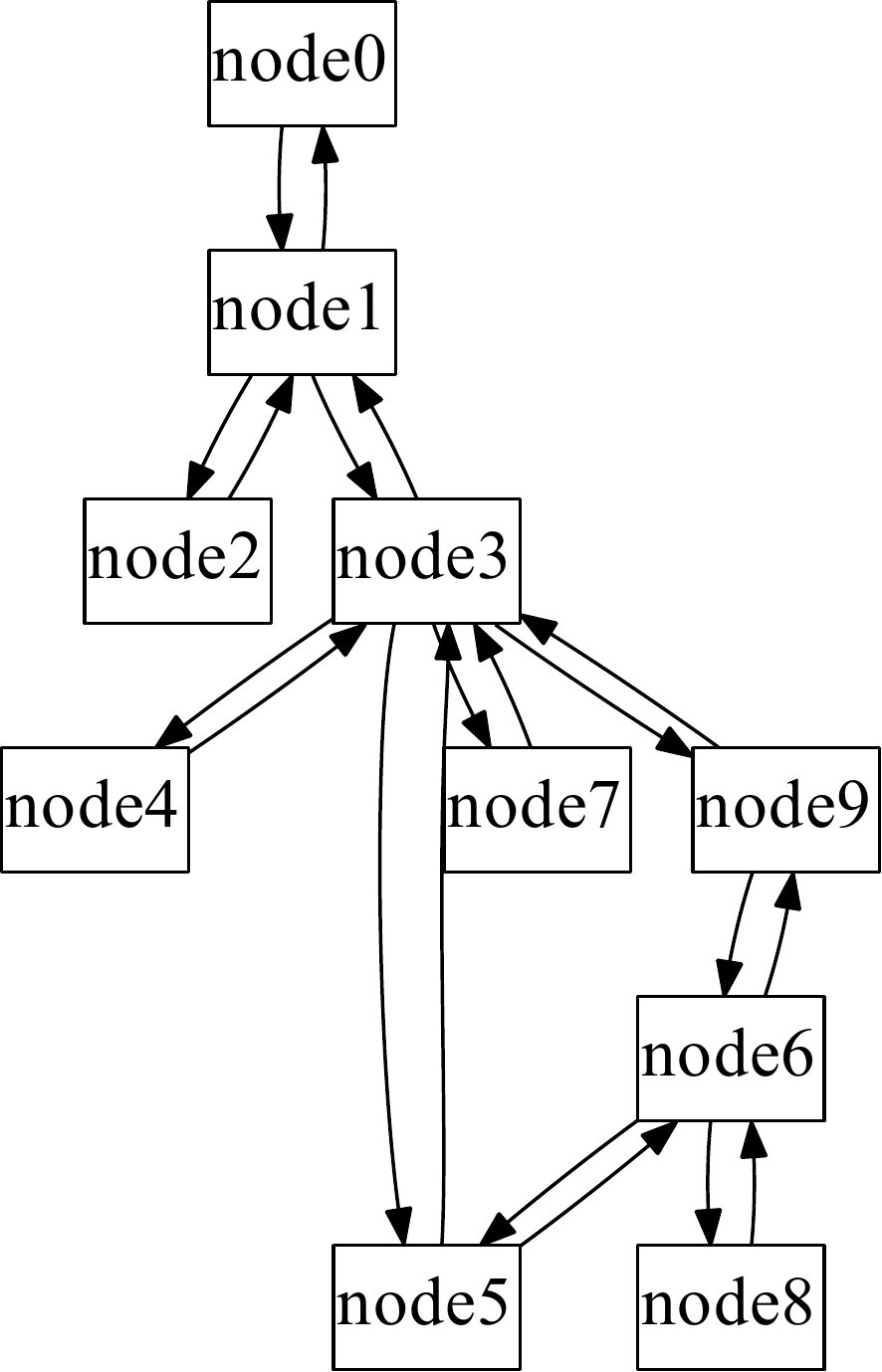}
\caption{Example substrate network}
\label{fig:graph}
\end{figure}

\begin{figure*}[tb]
\centering
\subfigure[\label{fig:state2}Initial embedding]{\includegraphics[height=54mm]{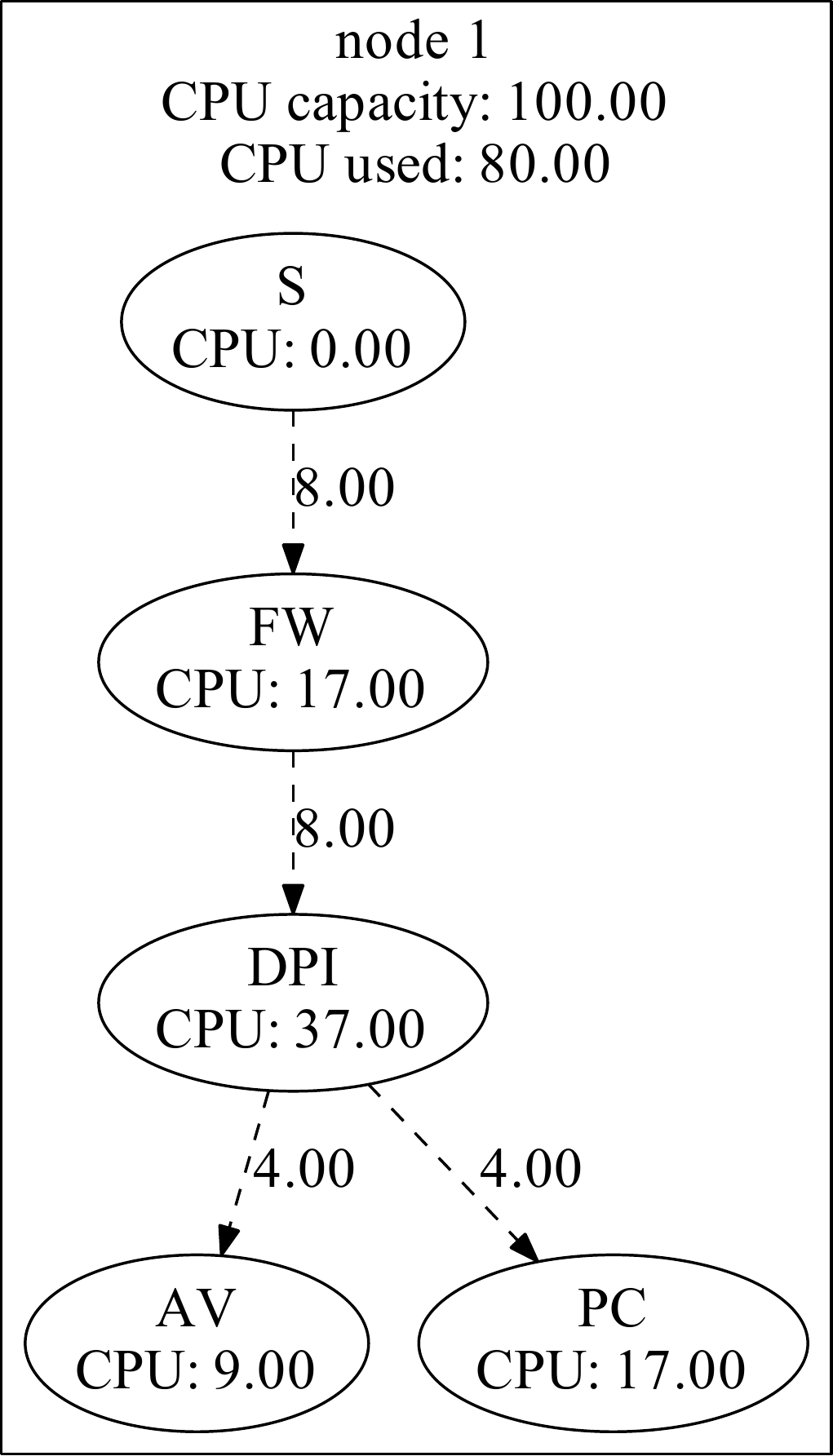}}
\hfil
\subfigure[\label{fig:state3}Result of increased source data rate]{\includegraphics[height=54mm]{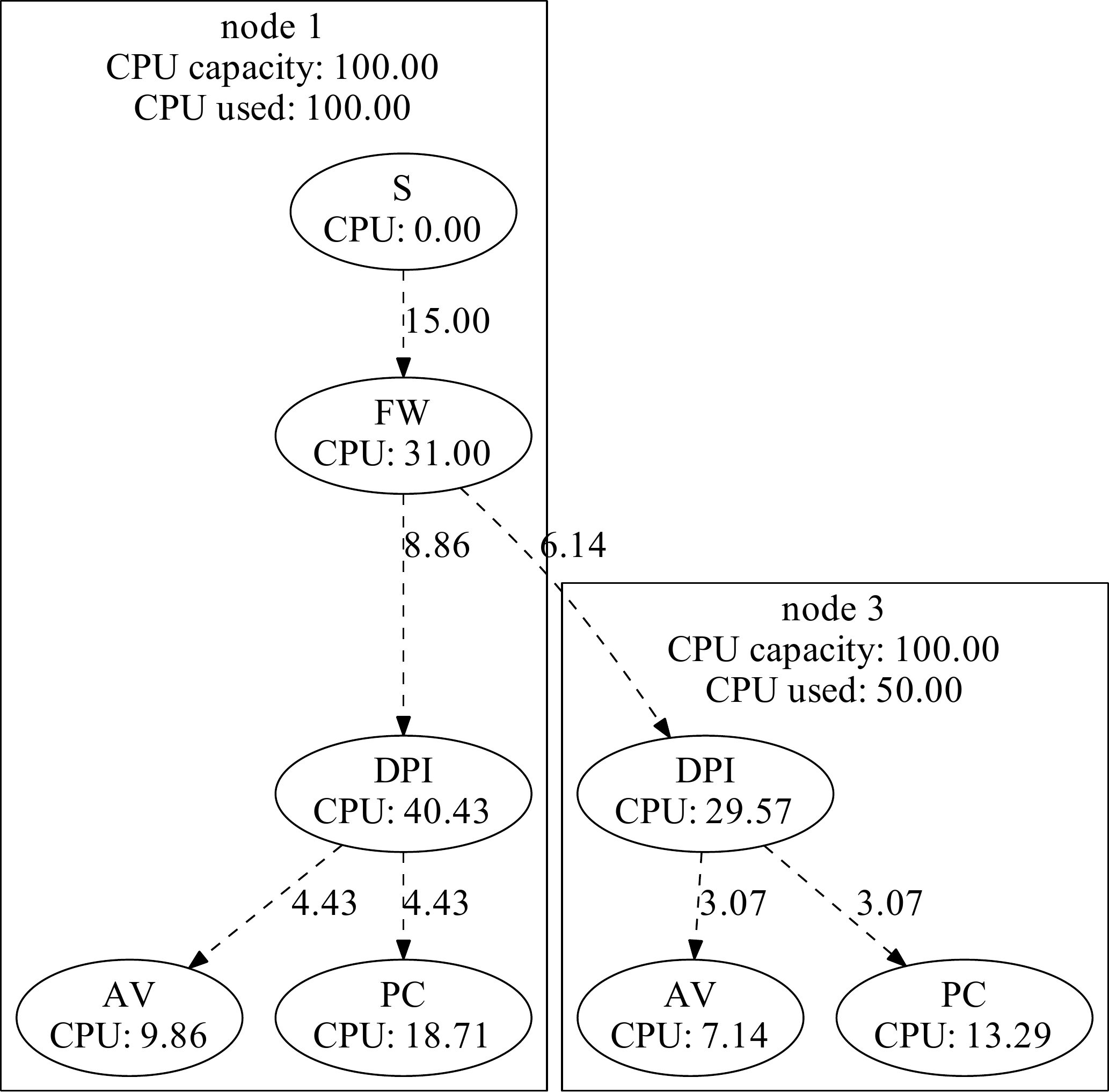}}
\hfil
\subfigure[\label{fig:state4}Result of the emergence of a second source]{\includegraphics[height=54mm]{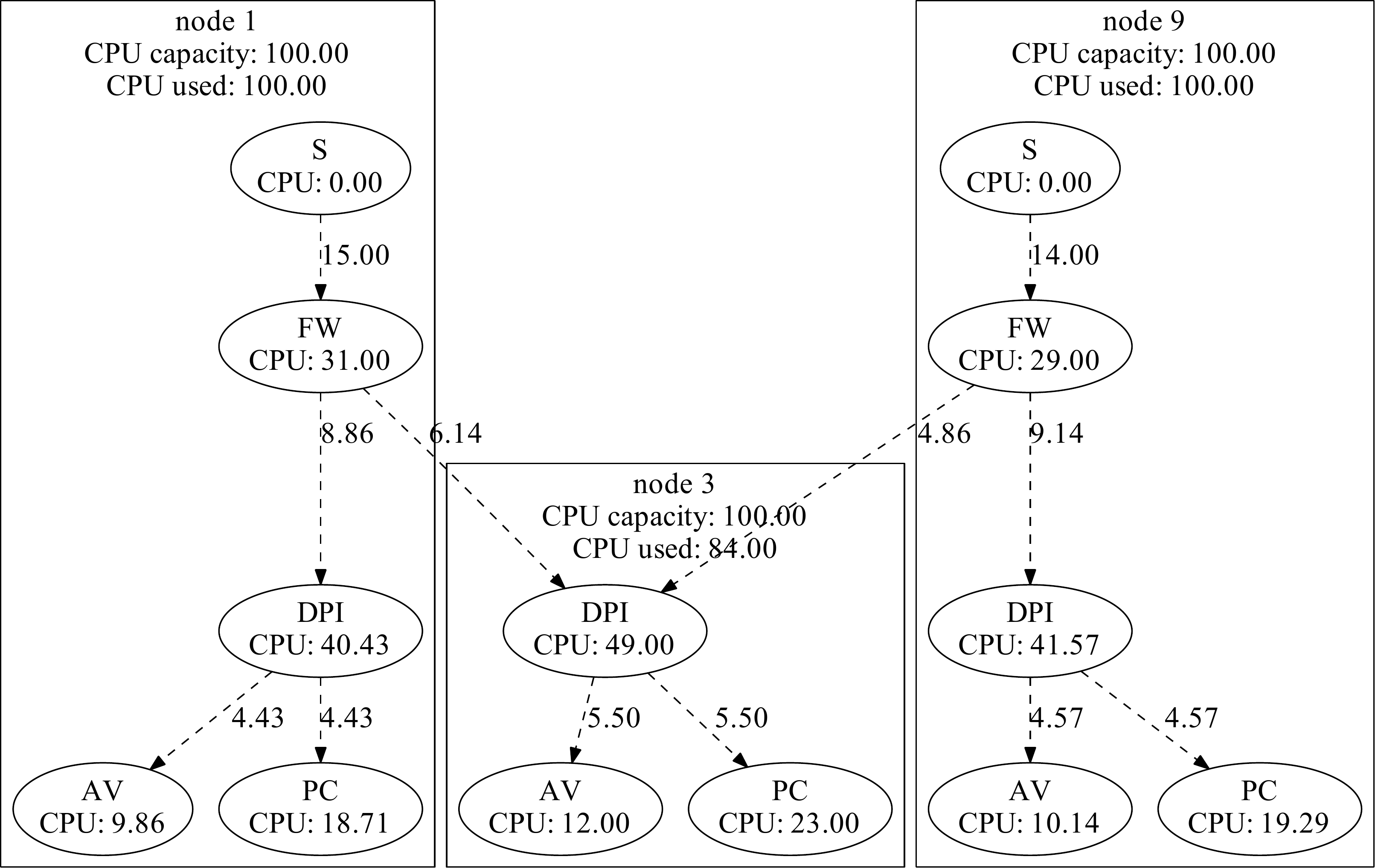}}
\caption{Illustrative example (memory values not shown for better readability)}
\end{figure*}

First, we illustrate our approach on a small substrate network of 10
nodes and 20 arcs (see Fig.~\ref{fig:graph}) in which the CPU and
memory capacity of each node is both 100. In this network, a service
consisting of a source (S), a firewall (FW), a deep packet inspection
(DPI) component, an anti-virus (AV) component, and a parental control (PC) component is deployed. Initially, there is a single source in node 1 with a moderate data rate. As a result, our algorithm
deploys all components of the service in node~1 (see Fig.~\ref{fig:state2}).

Subsequently, the data rate of the source increases. As a result, the resource demand of the processing components of the service increases so that they do not fit onto node~1 anymore. Our algorithm automatically re-scales the service by duplicating the DPI, AV, and PC components and automatically places the newly created instances on a nearby node, namely node~3 (see Fig.~\ref{fig:state3}).

Later on, a second source emerges for the same service on node~9. The algorithm automatically decides to create new processing component instances on node~9 to process as much as possible of the traffic of the new source locally. The excess traffic from the new FW instance that cannot be processed locally due to capacity constraints is routed to the existing DPI, AV, and PC instances on node~3 because node~3 still has sufficient free capacity (see Fig.~\ref{fig:state4}).

Already this small example shows the difficult trade-offs that template embedding involves. Next, we show that our approach is capable of handling also much more complex scenarios.

\subsection{Comparison of the algorithms}

\begin{figure}[tb]
\centering
\includegraphics[width=\columnwidth]{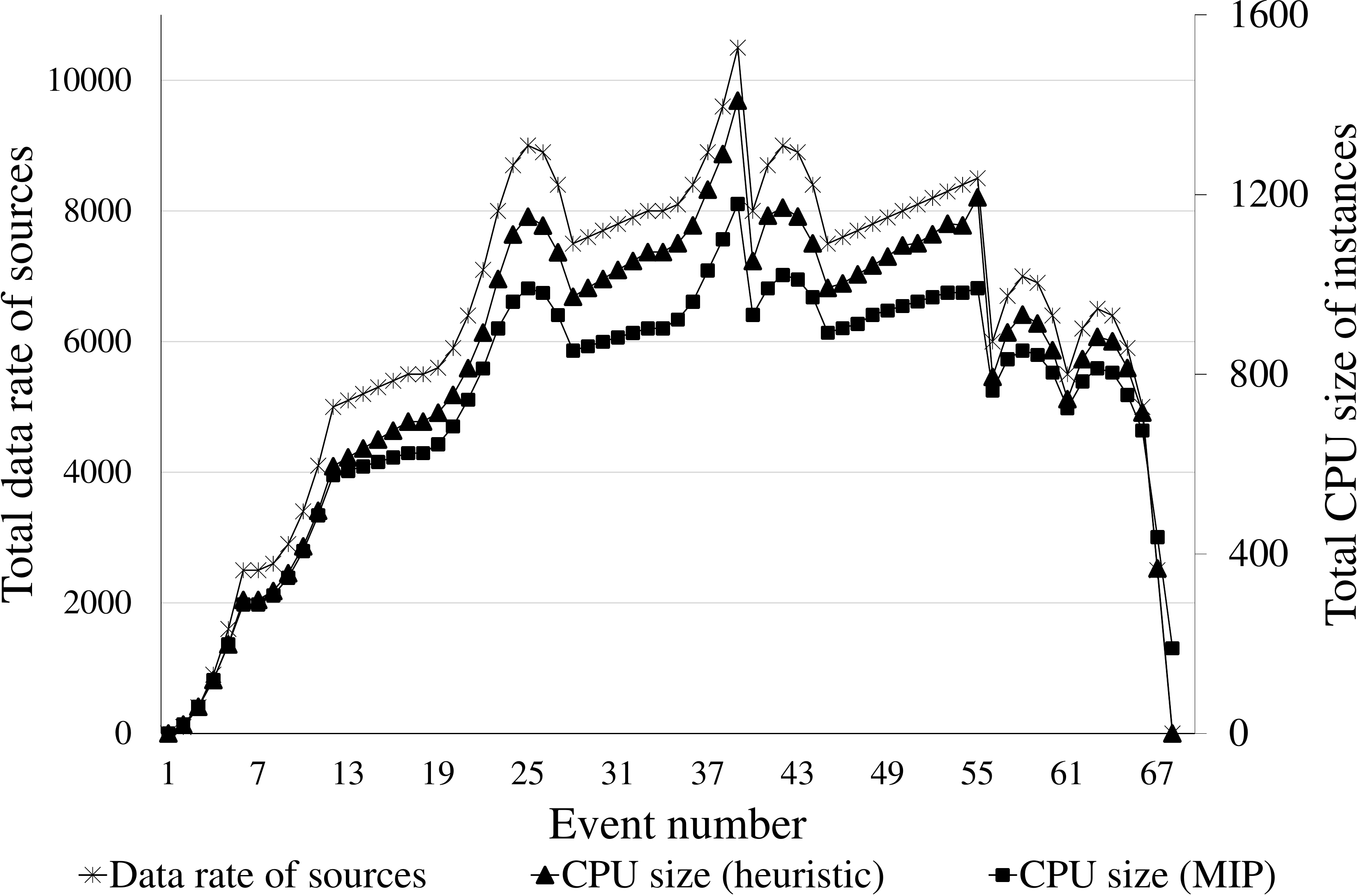}
\caption{Temporal development of the demand and the allocated capacity in a complex scenario}
\label{fig:scenario}
\end{figure}

\begin{figure}[tb]
\includegraphics[width=0.9\columnwidth]{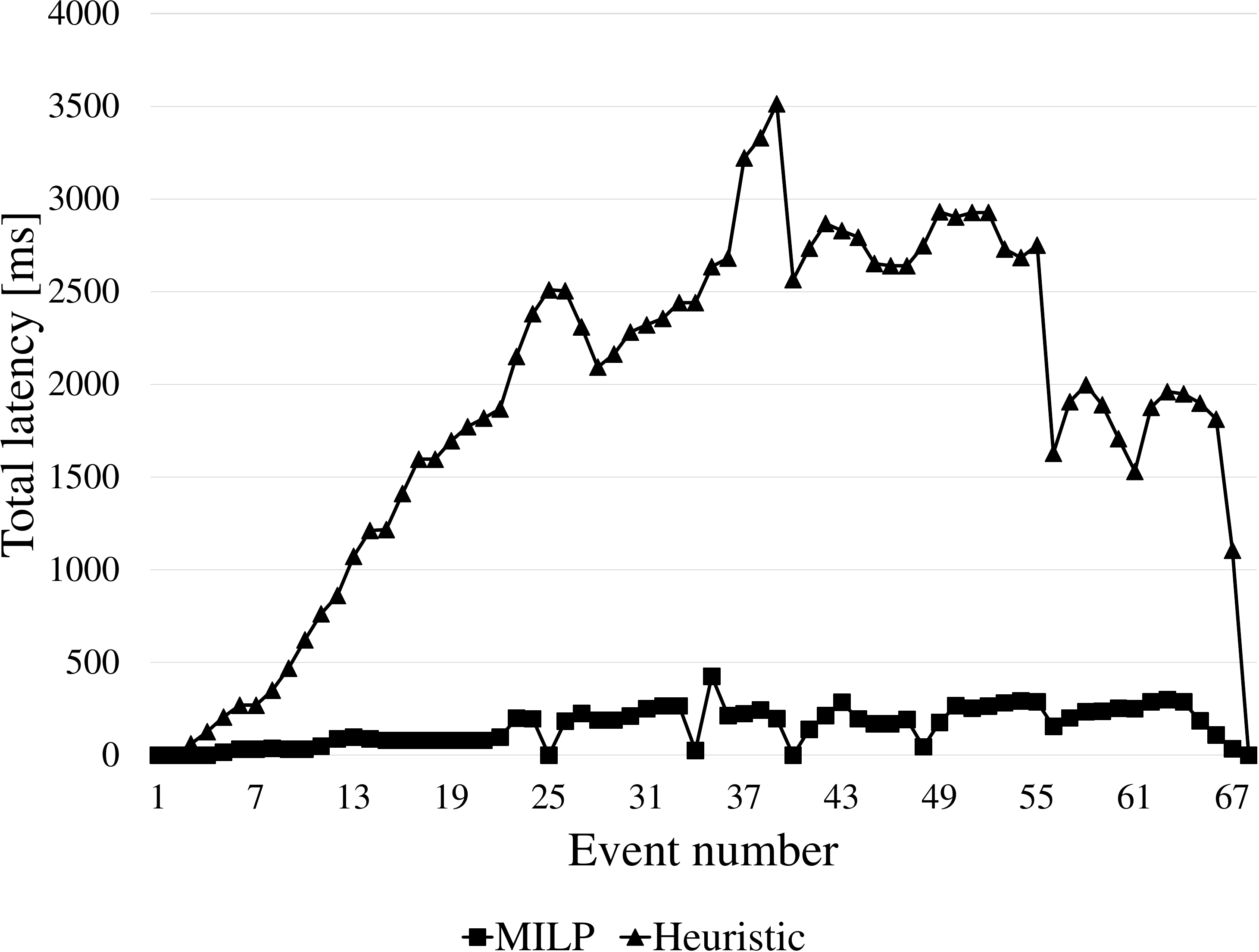}
\caption{Total latency over all created paths for the embedded template}
\label{fig:latency}
\end{figure}

We consider a substrate network with 20\,nodes and 44\,arcs, in which multiple services are deployed. Each service is a virtual content delivery network for video streaming, consisting of a streaming server, a DPI, a video optimizer, and a cache. The number of concurrently active services varies from 0 to 4, the number of sources varies from 0 to 20. Fig.~\ref{fig:scenario} shows how the total data rate of the sources (as a metric of the demand) and the total CPU size of the created instances (as a metric of the allocated processing capacity) change through re-optimization after each event. An event is the emergence or disappearance of a service, the emergence or disappearance of a source, or the change of the data rate of a source. As can be seen, the allocated capacity using both the heuristic and the MILP algorithms follow the demand very closely, meaning that our algorithms are successful in scaling the service in both directions to quickly react to changes in the demand. 

Regarding total data rate and total latency of the overlay edges, the
MILP algorithm performs better than the heuristic algorithm. For
example, Fig.~\ref{fig:latency} shows the total latency over all paths
created for the template in this scenario\footnote{In
  Fig.~\ref{fig:latency}, in the high-load area between event 20 and
  50, some problem instances are too complex to be solved within the
  60 seconds time limit we have set for the optimizer. This results in
  solutions with zero latency, as no paths are created.}. The reason
for this difference is that in the MILP algorithm, the optimal
location for all required instances can be determined at the same time. This results in
shorter distances between the source and the instances. The heuristic
algorithm, however, needs to create instances one by one, resulting in
larger data rates traveling over larger distances in the substrate network. 

In this scenario, to handle the peak demand, a total of 127 instances are created using the MILP algorithms, while the heuristic algorithm creates 261 instances. 

\subsection{Scalability}

\begin{figure}[tb]
\centering
\subfigure[\label{fig:exec_time_milp}Execution time of the MILP algorithm]{\includegraphics[width=0.9\columnwidth]{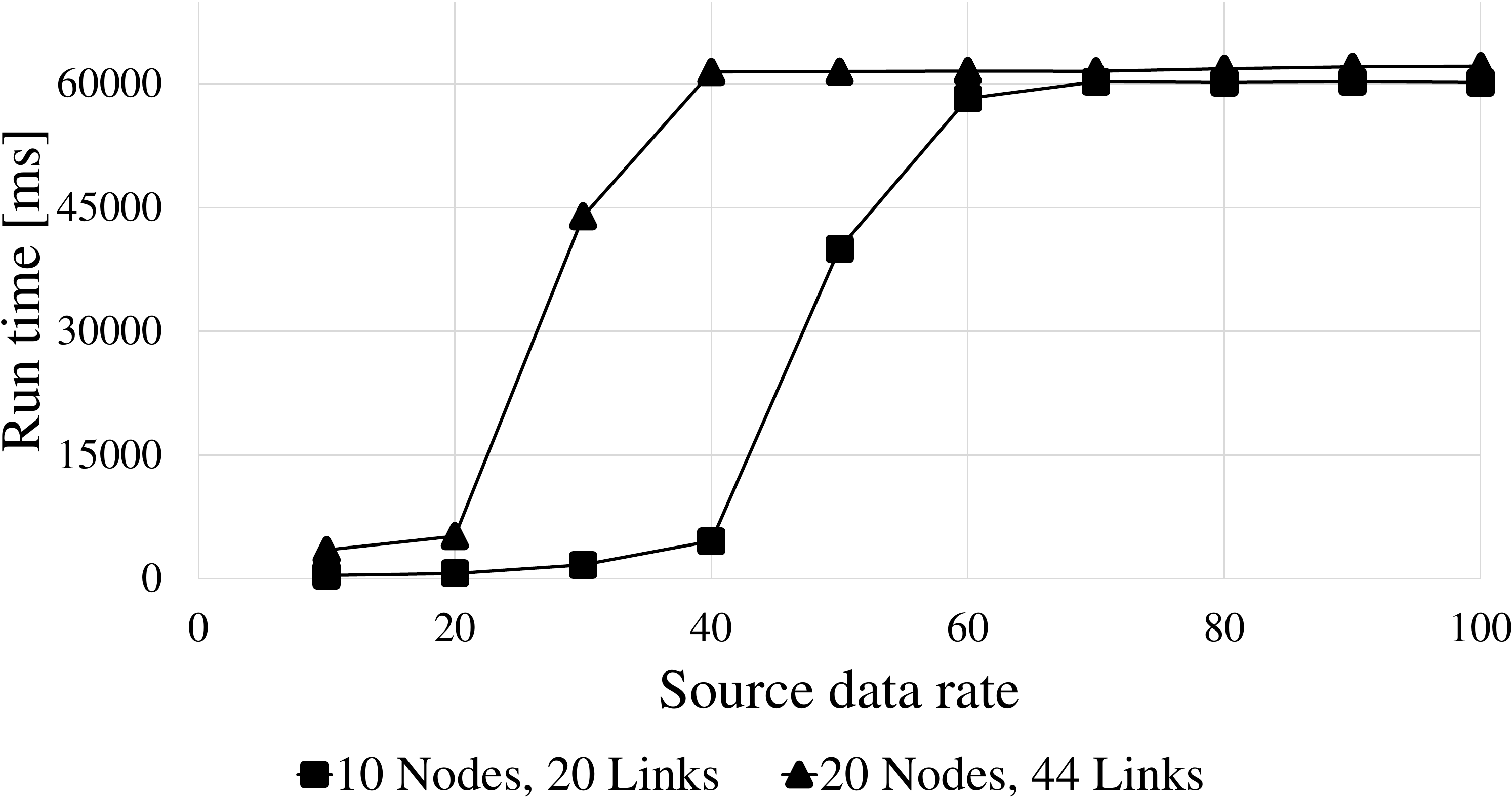}}
\subfigure[\label{fig:gap_milp}Optimality gap of the MILP algorithm]{\includegraphics[width=0.9\columnwidth]{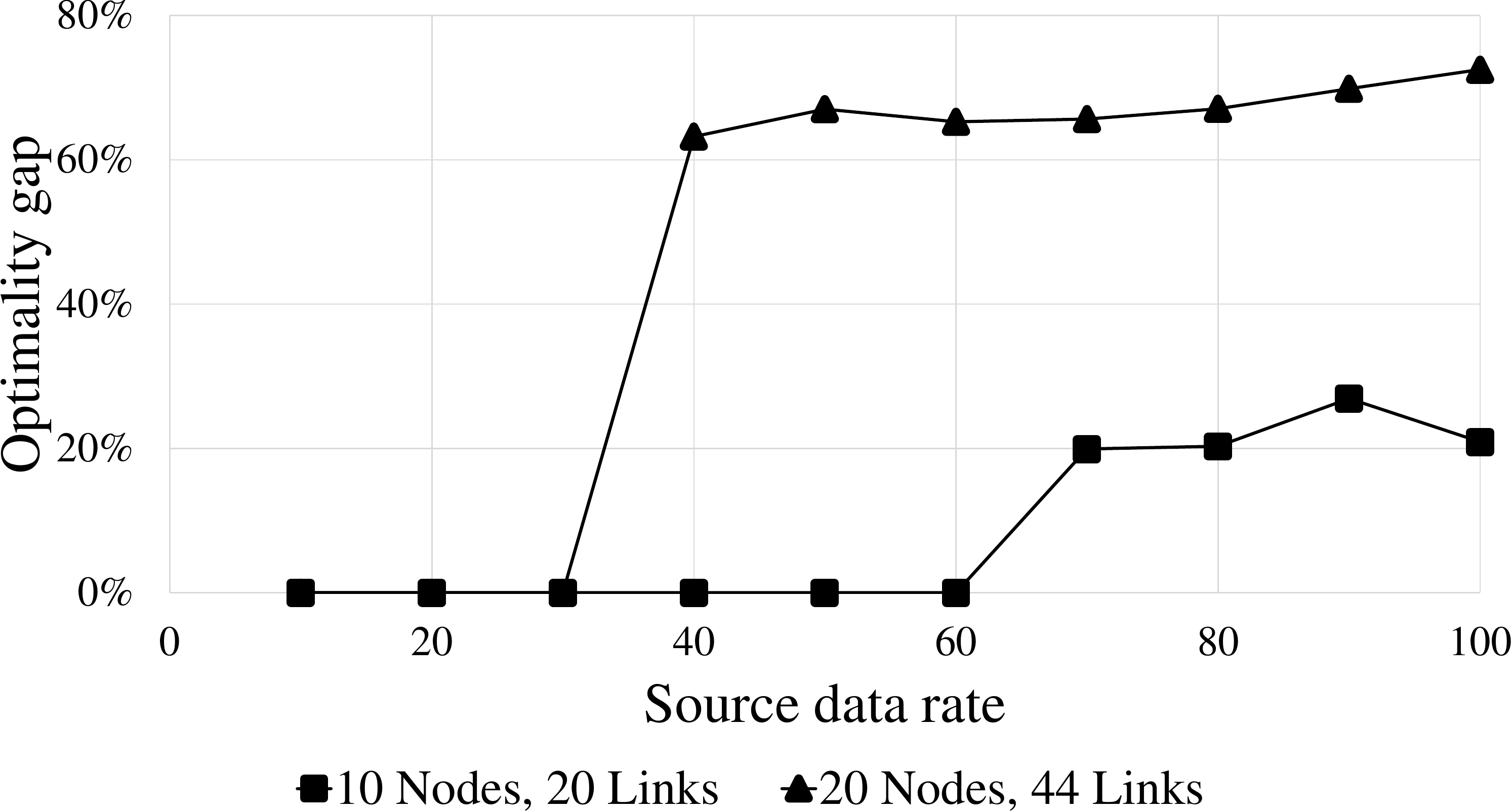}}
\subfigure[\label{fig:exec_time_heur}Execution time of the heuristic algorithm]{\includegraphics[width=0.9\columnwidth]{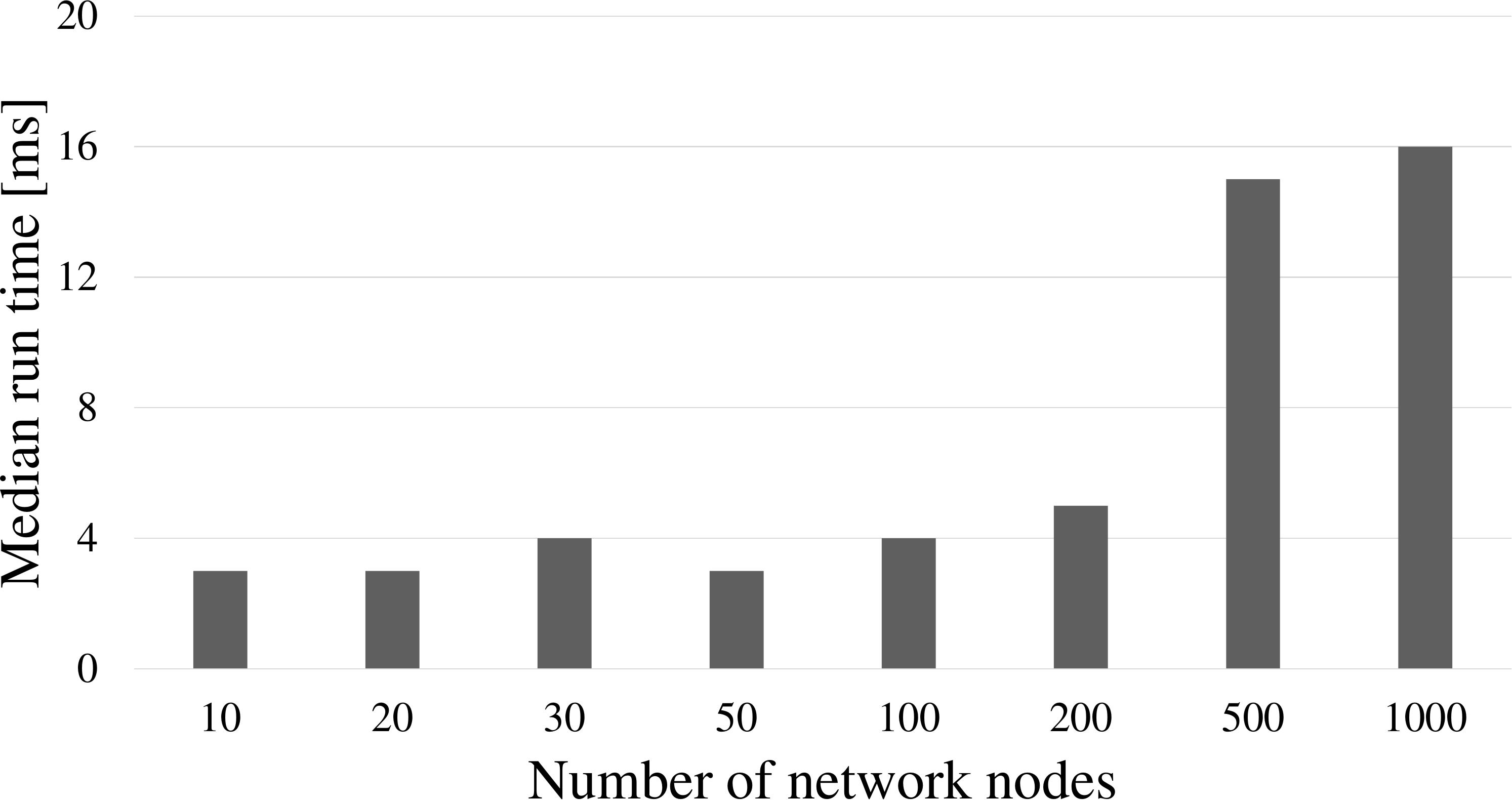}}
\caption{Scalability of the presented algorithms}
\end{figure}

Since the template embedding problem is NP-hard, it is foreseeable that the scalability of the MILP solver will be limited. In order to test this, we gradually increase the source data rate of the service from our first experiment, leading to an increasing number of instances; moreover, we also consider substrate networks of increasing size. In each case, the MILP solver is run with a time limit of 60 seconds, meaning that the solution process stops at (roughly) 60 seconds with the best solution and the best lower bound that the solver found until that time. The measurements were performed on a machine with Intel Core i5-4210U CPU @ 1.70GHz and 8GB RAM.

Fig.~\ref{fig:exec_time_milp} shows the execution time of the MILP
algorithm for different data rates and substrate network sizes, while
Fig.~\ref{fig:gap_milp} shows the corresponding gap between the found
solution and the lower bound. As can be seen, for a small network with
10 nodes and 20 arcs, the algorithm computes optimal results for the
lower half of source data rate values, and even for larger source data
rates, the optimality gap is quite low (around 20\,\%), meaning that
the results are almost optimal. However, for a bigger substrate
network with 20 nodes and 44 arcs, the solver reaches the time limit for much smaller source data rate and also the optimality gap is much bigger. For even bigger substrate networks, the performance of the algorithm further deteriorates, up to the point where it cannot be run anymore because of memory problems. The large sensitivity to the size of the substrate network is not surprising, given that the number of variables of the MILP is cubic in the size of the substrate network.

In contrast, as shown in Fig.~\ref{fig:exec_time_heur}, the execution
time of the heuristic algorithm remains very low even for the largest
substrate networks: for 1000 nodes and 2530 arcs, the execution time
is still below 20 milliseconds, rendering the heuristic practical for
real-world problem sizes as well.


%% file: conclusion.tex
\section{Conclusions}
\label{sec:concl}

We have presented JASPER, a fully automatic approach to scale, place,
and route multiple virtual network services on a common substrate
network. JASPER can be used for both the initial allocation of newly
requested services and the adaptation of existing services to changes
in the demand. Besides formally defining the problem and proving its
NP-hardness, we developed two algorithms for it, an MILP-based one and
a custom constructive heuristic.
Empiric tests have shown how our approach finds a balance between conflicting requirements and ensures that the allocated capacity quickly follows changes in the demand. The MILP-based algorithm gives optimal or near-optimal results for relatively small substrate network graphs, making it suitable for, e.g., calculations on top of a geographically distributed network where each node represents a data center. The heuristic remains very fast for even the largest networks that were tested. Overall, the tests gave evidence to the feasibility of our approach, which makes it possible (i) for service developers to specify services at a high level of abstraction and (ii) for providers to quickly re-optimize the system state after changes.

Promising future research directions include, beside further algorithmic enhancements to the presented algorithms and the development of new algorithms, the consideration of queuing incoming requests in the service components and the investigation of the effects of cyclic service templates.

%% file: paper.bbl
\begin{thebibliography}{10}
\providecommand{\url}[1]{#1}
\csname url@samestyle\endcsname
\providecommand{\newblock}{\relax}
\providecommand{\bibinfo}[2]{#2}
\providecommand{\BIBentrySTDinterwordspacing}{\spaceskip=0pt\relax}
\providecommand{\BIBentryALTinterwordstretchfactor}{4}
\providecommand{\BIBentryALTinterwordspacing}{\spaceskip=\fontdimen2\font plus
\BIBentryALTinterwordstretchfactor\fontdimen3\font minus
  \fontdimen4\font\relax}
\providecommand{\BIBforeignlanguage}[2]{{%
\expandafter\ifx\csname l@#1\endcsname\relax
\typeout{** WARNING: IEEEtran.bst: No hyphenation pattern has been}%
\typeout{** loaded for the language `#1'. Using the pattern for}%
\typeout{** the default language instead.}%
\else
\language=\csname l@#1\endcsname
\fi
#2}}
\providecommand{\BIBdecl}{\relax}
\BIBdecl

\bibitem{etsi-mano}
{ETSI NFV ISG}, ``{GS NFV-MAN 001 V1.1.1 Network Function Virtualisation (NFV);
  Management and Orchestration},'' Dec. 2014.

\bibitem{Fischer2013}
A.~Fischer, J.~F. Botero, M.~T. Beck, H.~de~Meer, and X.~Hesselbach, ``{Virtual
  Network Embedding: A Survey},'' \emph{IEEE Communications Surveys \&
  Tutorials}, vol.~15, no.~4, pp. 1888--1906, 2013.

\bibitem{houidi2015exact}
I.~Houidi, W.~Louati, and D.~Zeghlache, ``{Exact Multi-Objective Virtual
  Network Embedding in Cloud Environments},'' \emph{The Computer Journal},
  vol.~58, no.~3, pp. 403--415, 2015.

\bibitem{lorido2014review}
T.~Lorido-Botran, J.~Miguel-Alonso, and J.~A. Lozano, ``A review of
  auto-scaling techniques for elastic applications in cloud environments,''
  \emph{Journal of Grid Computing}, vol.~12, no.~4, pp. 559--592, 2014.

\bibitem{mann2016interplay}
Z.~A. Mann, ``Interplay of virtual machine selection and virtual machine
  placement,'' in \emph{Proceedings of the 5th European Conference on
  Service-Oriented and Cloud Computing}, 2016, pp. 137--151.

\bibitem{mann2015allocation}
------, ``Allocation of virtual machines in cloud data centers -- a survey of
  problem models and optimization algorithms,'' \emph{ACM Computing Surveys},
  vol.~48, no.~1, 2015.

\bibitem{divakaran2015towards}
D.~M. Divakaran and M.~Gurusamy, ``Towards flexible guarantees in clouds:
  Adaptive bandwidth allocation and pricing,'' \emph{IEEE Transactions on
  Parallel and Distributed Systems}, vol.~26, no.~6, pp. 1754--1764, 2015.

\bibitem{ahvar2015nacer}
E.~Ahvar, S.~Ahvar, N.~Crespi, J.~Garcia-Alfaro, and Z.~{\'A}. Mann, ``{NACER}:
  a network-aware cost-efficient resource allocation method for
  processing-intensive tasks in distributed clouds,'' in \emph{Proceedings of
  the 14th IEEE International Symposium on Network Computing and Applications},
  2015, pp. 90--97.

\bibitem{alicherry2013optimizing}
M.~Alicherry and T.~Lakshman, ``Optimizing data access latencies in cloud
  systems by intelligent virtual machine placement,'' in \emph{Proceedings of
  IEEE Infocom}, 2013, pp. 647--655.

\bibitem{ahvar2016cacev}
E.~Ahvar, S.~Ahvar, Z.~A. Mann, N.~Crespi, J.~Garcia-Alfaro, and R.~Glitho,
  ``{CACEV}: a cost and carbon emission-efficient virtual machine placement
  method for green distributed clouds,'' in \emph{IEEE 13th International
  Conference on Services Computing}, 2016, pp. 275--282.

\bibitem{Bellavista2015}
P.~Bellavista, F.~Callegati, W.~Cerroni, C.~Contoli, A.~Corradi, L.~Foschini,
  A.~Pernafini, and G.~Santandrea, ``{Virtual network function embedding in
  real cloud environments},'' \emph{Computer Networks}, vol.~93, pp. 506--517,
  dec 2015.

\bibitem{Wang2017}
X.~Wang, C.~Wu, F.~Le, A.~Liu, Z.~Li, and F.~Lau, ``{Online VNF scaling in
  datacenters},'' in \emph{IEEE International Conference on Cloud Computing,
  CLOUD}.\hskip 1em plus 0.5em minus 0.4em\relax IEEE, jun 2017, pp. 140--147.

\bibitem{Keller2014b}
M.~Keller, C.~Robbert, and H.~Karl, ``{Template Embedding: Using Application
  Architecture to Allocate Resources in Distributed Clouds},'' in
  \emph{IEEE/ACM 7th International Conference on Utility and Cloud Computing
  (UCC)}, 2014.

\bibitem{herrera2016resource}
J.~G. Herrera and J.~F. Botero, ``Resource allocation in {NFV}: A comprehensive
  survey,'' \emph{IEEE Transactions on Network and Service Management},
  vol.~13, no.~3, pp. 518--532, 2016.

\bibitem{Kuo2016}
T.~W. Kuo, B.~H. Liou, K.~C.~J. Lin, and M.~J. Tsai, ``Deploying chains of
  virtual network functions: On the relation between link and server usage,''
  in \emph{IEEE INFOCOM}, 2016.

\bibitem{Ahvar2017}
S.~Ahvar, H.~P. Phyu, and R.~Glitho, ``{CCVP: Cost-efficient Centrality-based
  VNF Placement and Chaining Algorithm for Network Service Provisioning},'' in
  \emph{IEEE NetSoft}.\hskip 1em plus 0.5em minus 0.4em\relax IEEE, jul 2017,
  pp. 1--9.

\bibitem{bari2016orchestrating}
F.~Bari, S.~R. Chowdhury, R.~Ahmed, R.~Boutaba, and O.~C. M.~B. Duarte,
  ``Orchestrating virtualized network functions,'' \emph{IEEE TNSM}, vol.~13,
  no.~4, pp. 725--739, 2016.

\bibitem{Khebbache2017}
S.~Khebbache, M.~Hadji, and D.~Zeghlache, ``{Virtualized network functions
  chaining and routing algorithms},'' \emph{Computer Networks}, vol. 114, pp.
  95--110, feb 2017.

\bibitem{Luizelli2017}
M.~C. Luizelli, W.~L. {da Costa Cordeiro}, L.~S. Buriol, and L.~P. Gaspary,
  ``{A fix-and-optimize approach for efficient and large scale virtual network
  function placement and chaining},'' \emph{Computer Communications}, vol. 102,
  pp. 67--77, apr 2017.

\bibitem{draexler2017ijnm}
S.~Dr\"axler and H.~Karl, ``Specification, composition, and placement of
  network services with flexible structures,'' \emph{International Journal of
  Network Management}, vol.~27, no.~2, pp. e1963--n/a, 2017, e1963 nem.1963.

\bibitem{Beck2015}
M.~T. Beck and J.~F. Botero, ``{Scalable and coordinated allocation of service
  function chains},'' \emph{Computer Communications}, vol.~0, pp. 1--11, apr
  2015.

\bibitem{mehraghdam-netsoft16}
S.~Mehraghdam and H.~Karl, ``{Placement of Services with Flexible Structures
  Specified by a YANG Data Model},'' in \emph{IEEE 2nd Conference on Network
  Softwarization (NetSoft)}, 2016.

\bibitem{sahhaf2015}
S.~Sahhaf, W.~Tavernier, D.~Colle, and M.~Pickavet, ``{Network Service Chaining
  with Efficient Network Function Mapping Based on Service Decompositions},''
  in \emph{{IEEE 1st Conference on Network Softwarization (NetSoft)}}, April
  2015.

\bibitem{moens2014vnf}
H.~Moens and F.~De~Turck, ``{VNF-P: A Model for Efficient Placement of
  Virtualized Network Functions},'' in \emph{{IEEE 10th Conference on Network
  and Service Management (CNSM)}}, 2014.

\bibitem{savi2015impact}
M.~Savi, M.~Tornatore, and G.~Verticale, ``{Impact of Processing Costs on
  Service Chain Placement in Network Functions Virtualization},'' in \emph{IEEE
  1st Conference on Network Function Virtualization and Software Defined
  Network (NFV-SDN)}, 2015.

\bibitem{mijumbidesign}
R.~Mijumbi, J.~Serrat, J.-L. Gorricho, N.~Bouten, F.~De~Turck, and S.~Davy,
  ``{Design and Evaluation of Algorithms for Mapping and Scheduling of Virtual
  Network Functions},'' in \emph{{IEEE 1st Conference on Network Softwarization
  (NetSoft)}}, 2015.

\bibitem{beck2015coordinated}
M.~T. Beck and J.~F. Botero, ``{Coordinated Allocation of Service Function
  Chains},'' in \emph{IEEE Global Communications Conference}, 2015.

\bibitem{draexler2017joint}
S.~Dr{\"a}xler, H.~Karl, and Z.~{\'A}. Mann, ``Joint optimization of scaling
  and placement of virtual network services,'' in \emph{17th IEEE/ACM
  International Symposium on Cluster, Cloud and Grid Computing (CCGrid 2017)},
  2017, pp. 365--370.

\bibitem{profiling2016}
M.~Peuster and H.~Karl, ``{Understand Your Chains: Towards Performance
  Profile-based Network Service Management},'' in \emph{Proceeding of the Fifth
  European Workshop on Software Defined Networks}.\hskip 1em plus 0.5em minus
  0.4em\relax IEEE, 2016.

\bibitem{osmwebsite}
``{OSM},'' \url{https://osm.etsi.org}, date accessed: 2017-09-29.

\bibitem{sonatawebsite}
``{SONATA},'' \url{http://sonata-nfv.eu}, date accessed: 2017-09-29.

\bibitem{unifywebsite}
``{UNIFY},'' \url{www.fp7-unify.eu}, date accessed: 2017-09-29.

\bibitem{sonata-paper}
S.~Dr\"axler, H.~Karl, M.~Peuster, H.~R. Kouchaksaraei, M.~Bredel, J.~Lessmann,
  T.~Soenen, W.~Tavernier, S.~Mendel-Brin, and G.~Xilouris, ``Sonata: Service
  programming and orchestration for virtualized software networks,'' in
  \emph{2017 IEEE International Conference on Communications Workshops (ICC
  Workshops)}, May 2017, pp. 973--978.

\bibitem{karp1972reducibility}
R.~M. Karp, ``Reducibility among combinatorial problems,'' in \emph{Complexity
  of Computer Computations}, R.~E. Miller and J.~W. Thatcher, Eds., 1972, pp.
  85--103.

\bibitem{hochbaum2008pseudoflow}
D.~S. Hochbaum, ``The pseudoflow algorithm: A new algorithm for the
  maximum-flow problem,'' \emph{Operations Research}, vol.~56, no.~4, pp.
  992--1009, 2008.

\bibitem{korf1993linear}
R.~E. Korf, ``Linear-space best-first search,'' \emph{Artificial Intelligence},
  vol.~62, no.~1, pp. 41--78, 1993.

\bibitem{infuhr2013solving}
J.~Inf{\"u}hr and G.~R. Raidl, ``Solving the virtual network mapping problem
  with construction heuristics, local search and variable neighborhood
  descent,'' in \emph{Proceedings of the 13th European Conference on
  Evolutionary Computation in Combinatorial Optimization}, 2013, pp. 250--261.

\bibitem{draft-liu-sfc-use-cases-08}
W.~Liu, H.~Li, O.~Huang, M.~Boucadair, N.~Leymann, Q.~Fu, Q.~Sun, C.~Pham,
  C.~Huang, J.~Zhu, and P.~He, ``{Service Function Chaining (SFC) General Use
  Cases},'' {Work in progress}, {IETF Secretariat}, {Internet-Draft}
  draft-liu-sfc-use-cases-08, Sep. 2014.

\end{thebibliography}
